%% file: resubmit.tex
\documentclass[superscriptaddress,10pt,prx,twocolumn, nofootinbib,showpacs,longbibliography,float fix]{revtex4-2}
\usepackage{blindtext}
\usepackage{lipsum}
\usepackage{graphics}
\usepackage{amsmath}
\usepackage{graphicx}
\usepackage{graphics}
\usepackage{amssymb}
\usepackage{amsthm}
\usepackage{verbatim}
\usepackage{physics}
\usepackage{float}
\usepackage{mathtools}
\usepackage{microtype}
\usepackage[normalem]{ulem}
\usepackage[colorlinks]{hyperref}
\usepackage[dvipsnames]{xcolor}
\definecolor{darkblue}{rgb}{0.,0.,0.4}

\def\bea{\begin{eqnarray}}
\def\eea{\end{eqnarray}}
\def\nn{\nonumber}
\usepackage{array}
\usepackage{soul}
\usepackage{tablefootnote}

\def\triv{\textsf{inv}}

\def\sgn{\textsf{sgn}}
\def\twod{\textsf{2d}}

\newtheorem{result}{Result}
\newtheorem{definition}{Definition}[section]
\newtheorem{theorem}{Theorem}[section]

\newtheorem{lemma}[theorem]{Lemma} 

\usepackage{ytableau}
\begin{document}

\title{Symmetry enforced entanglement in maximally mixed states}

\author{Amin Moharramipour}
\affiliation{Department of Physics, University of Toronto,
60 St. George Street, Toronto, Ontario, M5S 1A7, Canada}
\affiliation{%
 Perimeter Institute for Theoretical Physics, Waterloo, Ontario N2L 2Y5, Canada}
 
\author{Leonardo A. Lessa}

\affiliation{%
 Perimeter Institute for Theoretical Physics, Waterloo, Ontario N2L 2Y5, Canada}

 \affiliation{Department of Physics and Astronomy, University of Waterloo, Waterloo, Ontario N2L 3G1, Canada}

\author{Chong Wang}
\author{Timothy H. Hsieh}
\author{Subhayan Sahu}
\email{ssahu@perimeterinstitute.ca}
\affiliation{%
 Perimeter Institute for Theoretical Physics, Waterloo, Ontario N2L 2Y5, Canada}%
\begin{abstract}

    Entanglement in quantum many-body systems is typically fragile to interactions with the environment. Generic unital quantum channels, for example, have the maximally mixed state with no entanglement as their unique steady state. However, we find that for a unital quantum channel that is `strongly symmetric', i.e. it preserves a global on-site symmetry, the maximally mixed steady state in certain symmetry sectors can be highly entangled. For a given symmetry, we analyze the entanglement and correlations of the maximally mixed state in the invariant sector (MMIS), and show that the entanglement of formation and distillation are exactly computable and equal for any bipartition.  For all Abelian symmetries, the MMIS is separable, and for all non-Abelian symmetries, the MMIS is entangled. Remarkably, for non-Abelian continuous symmetries described by compact semisimple Lie groups (e.g. $SU(2)$), the bipartite entanglement of formation for the MMIS scales logarithmically $\sim \log N$ with the number of qudits $N$.

\end{abstract}

\maketitle
\tableofcontents

\section{Introduction}


Entanglement in quantum many-body systems is a powerful resource for quantum information processing. Entanglement also provides a useful measure of quantum correlations in a system, and is essential for diagnosing quantum phases of matter and the phase transitions between them. However, entanglement is generically believed to be fragile to external noise (i.e. when connected to an external bath), except at sufficiently low temperatures in equilibrium, or by active steering and non-equilibrium processes.

\begin{figure*}
    \centering
    \def\svgwidth{\textwidth}
   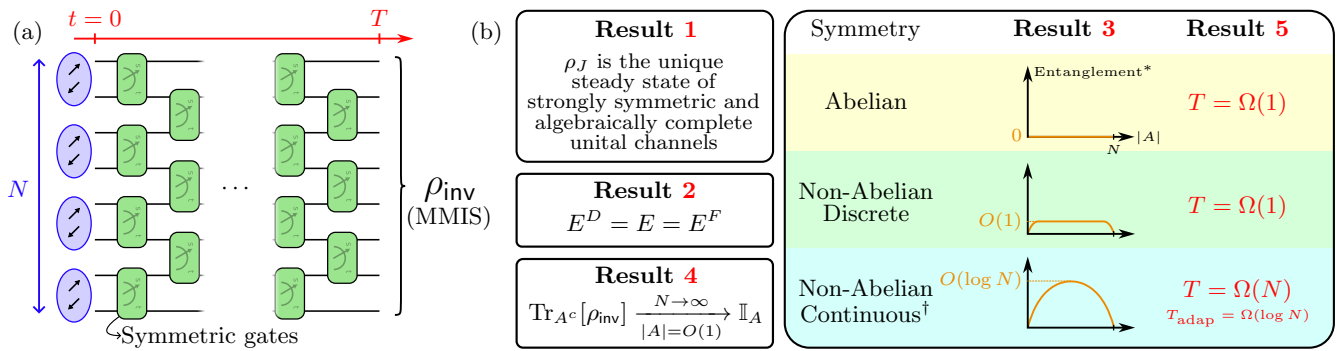
    \caption{ (a) 
    \textbf{Setting}: An example of a symmetric channel obeying conditions of Result \ref{res:SS} that acts on a symmetric product-like state. Each element of the brickwork circuit is a local channel which preserves a global symmetry (for example, singlet-triplet measurement channel for $SU(2)$ on a qubit chain). We show that the steady state of such a channel is the maximally mixed invariant state (MMIS) $\rho_\triv$. We highlight that the results listed here do not depend on the channel being brickwork; any unital symmetric channel that produces the maximally mixed state within a particular symmetry sector will lead to the same behavior of entanglement. (b) \textbf{Summary of results} (Sec. \ref{sec:main_results}): We first prove \textbf{Result \ref{res:SS}}, that the steady state of such evolution is the maximally mixed state $\rho_J$ in the same symmetry sector $J$. If the initial state transforms trivially under the symmetry, such as the neighbouring singlets illustrated for $SU(2)$ (See Sec.~\ref{sec:warmup}), then the steady-state is the maximally mixed invariant state $\rho_\triv$, whose properties are explored in Results~\ref{res:Ent}-\ref{res:time}. In \textbf{Result~\ref{res:Ent}}, $E^D$ and $E^F$ mean entanglement of distillation and formation respectively, $E$ is any other \textit{good} entanglement measure (See Sec. \ref{sec:entanglement_measures}), and the equality holds for the MMIS with respect to any bipartition. In \textbf{Result~\ref{res:exp_values}}, region $A$ is assumed to have size $|A| = O(1)$ not scaling with $N$, and $\mathbb{I}_A$ means the maximally mixed state in $A$.  For \textbf{Results~\ref{res:EBC}} and \textbf{\ref{res:time}}, $(\dagger)$ refer to compact semisimple non-Abelian Lie groups, and $(*)$ entanglement measures include formation, distillation, and logarithmic negativity. Here we quote bipartite entanglement of the MMIS in the thermodynamic limit $N\to\infty$ (Result \ref{res:EBC}), and the times $T$ and $T_{\text{adap}} \leq T$ required by local and adaptive (with measurements and feedback) channels to prepare such a state, respectively (Result \ref{res:time}).}
    \label{fig:firstfig}
\end{figure*}

For instance, when the system is weakly connected to a bath at finite temperature, the stationary state at equilibrium is typically the Gibbs state corresponding to the system Hamiltonian. Long-range quantum entanglement can survive in the Gibbs state at sufficiently low temperatures for some geometrically non-local Hamiltonians with enhanced quantum error-correcting properties~\cite{freedman_quantum_2014}, and topological orders in four and higher dimensions~\cite{dennis_topological_2002, alicki_thermal_2008}. However, at high temperatures, Gibbs states of local Hamiltonians become unentangled~\cite{bakshi_high-temperature_2024}. In general, for characterising mixed state entanglement, one has to establish the non separability of the mixed state which is generally hard.

In the active non-equilibrium setting, however, open quantum systems can be designed such that the steady states are entangled. In fact, for any pure quantum state, which may be highly entangled, one can design a noise channel which produces it as a stationary state~\cite{verstraete_quantum_2009, kraus_preparation_2008}. More generally, interesting and non trivial mixed quantum states can also be prepared by active feedback and measurements~\cite{lu_mixed-state_2023, lee2022decoding, zhu2022nishimoris, LessaChengWang}. Several works ~\cite{bao2023mixedstate, fan2024diagnostics, lee2023quantum, Zou_2023} have also identified that long-range entanglement also survives in the mixed state generated by the application of short depth local noise channels on topologically ordered  or quantum critical states. 

In this work, we investigate the interplay between symmetry and entanglement in states that are maximally mixed in a given symmetry sector. Symmetries are manifested in mixed states in two inequivalent ways~\cite{Buca_2012, Albert_2014, albert2018lindbladians}: $\rho$ is said to have a strong symmetry (also called exact symmetry) if it is a mixture of pure states carrying the same charge under the symmetry. Alternatively, the symmetry can be weak (or average), and not strong, such that it is a mixture of states with well-defined, but not equal, charges. 

We find that quantum channels that are strongly symmetric with an on-site non-Abelian symmetry have steady states (within certain symmetry sectors) that are generically entangled. Recently, it was identified that strong anomalous symmetry, whether zero-form~\cite{LessaChengWang,WangLi2024} or higher-form~\cite{ ellison2024classification, sohal2024noisy,wang2024intrinsic,LiLeeYoshida,
Lessatoappear}, guarantees long-range entanglement in pure or mixed states. Similarly, we find symmetry-enforced entanglement in the maximally mixed symmetric states, but, in contrast to these works, our results demonstrate how strong on-site zero-form non-Abelian symmetry can constrain not only the structure but also the amount (e.g. $\log N$ vs. $O(1)$) of entanglement in otherwise maximally mixed states.



For generic unital quantum channels (i.e channels for which the identity operator is a stationary state), such as those generated by measurements without nonlocal feedback, the stationary state is not expected to be long-range entangled. In the presence of just passive measurements, the state is effectively dephased, giving rise to a classical mixture of trajectories which generally do not possess genuine mixed state entanglement. This is the motivation behind active error correction in quantum error correction protocol, where measurements are followed by active feedback. A similar issue arises in the study of measurement-induced phase transitions~\cite{li_measurement-driven_2019, Skinner_2019}. The entanglement phase transition in unitary and measurement dynamics is only visible in individual trajectories, and not in the mixed state described by the mixture of the trajectories. In fact, in the absence of symmetries, the steady state is typically the featureless maximally mixed `identity' state, which is completely separable with no classical correlations as well. 

To motivate the main results, we first describe a simple numerical study of a measurement channel that leads to highly entangled stationary mixed state, without any active feedback. 


\subsection{A warmup example}\label{sec:warmup}

Consider a chain of even number $N$ of spin-$1/2$ qubits, initialized in a dimer state of nearest neighbor singlets (assuming periodic boundary condition), $\otimes_{i = \text{even}}\ket{s}_{i,i+1}$, where $\ket{s}_{i,i+1} = \frac{1}{\sqrt{2}}\left(\ket{0_{i}1_{i+1}}-\ket{1_{i}0_{i+1}}\right)$. This state has total spin $0$, and is thus a singlet under the global $SU(2)$ symmetry of the spin chain. We perform measurements on this state in a brickwork fashion: at every elementary step one measures the pair of nearest neighbor qubits in the singlet-triplet basis. The local projective  measurement is described by the projector on the singlet, $P_{s,i} = \ket{s}_{i,i+1}\bra{s}_{i,i+1}$, and the triplets, $P_{t,i} = \mathbb{I}-P_{s,i}$ (See Fig. \ref{fig:firstfig}(a)). After each elementary measurement step, the density matrix describing the whole state undergoes evolution under the measurement channel, $\rho \to P_{s}\rho P_{s}^{\dagger}+P_{t}\rho P_{t}^{\dagger}$. One time step of such a circuit is when all the nearest neighbors have been measured once. We are interested in characterising the entanglement of the mixed state as it evolves. We calculate the logarithmic negativity, a computable measure of mixed state entanglement ~\cite{vidal_computable_2002}, of half of the spin chain. The logarithmic negativity for a bipartition $A:B$ is defined as, $E^{\mathcal{N}}_{A:B}(\rho) = \log || \rho^{T_{A}}||_{1}$, where $T_{A}$ refers to the partial transposition with respect to region $A$, and $|| A ||_{1} \equiv \Tr \sqrt{A^\dagger A}$ . 

\begin{figure}
    \centering
    \includegraphics[width = 0.9\columnwidth]{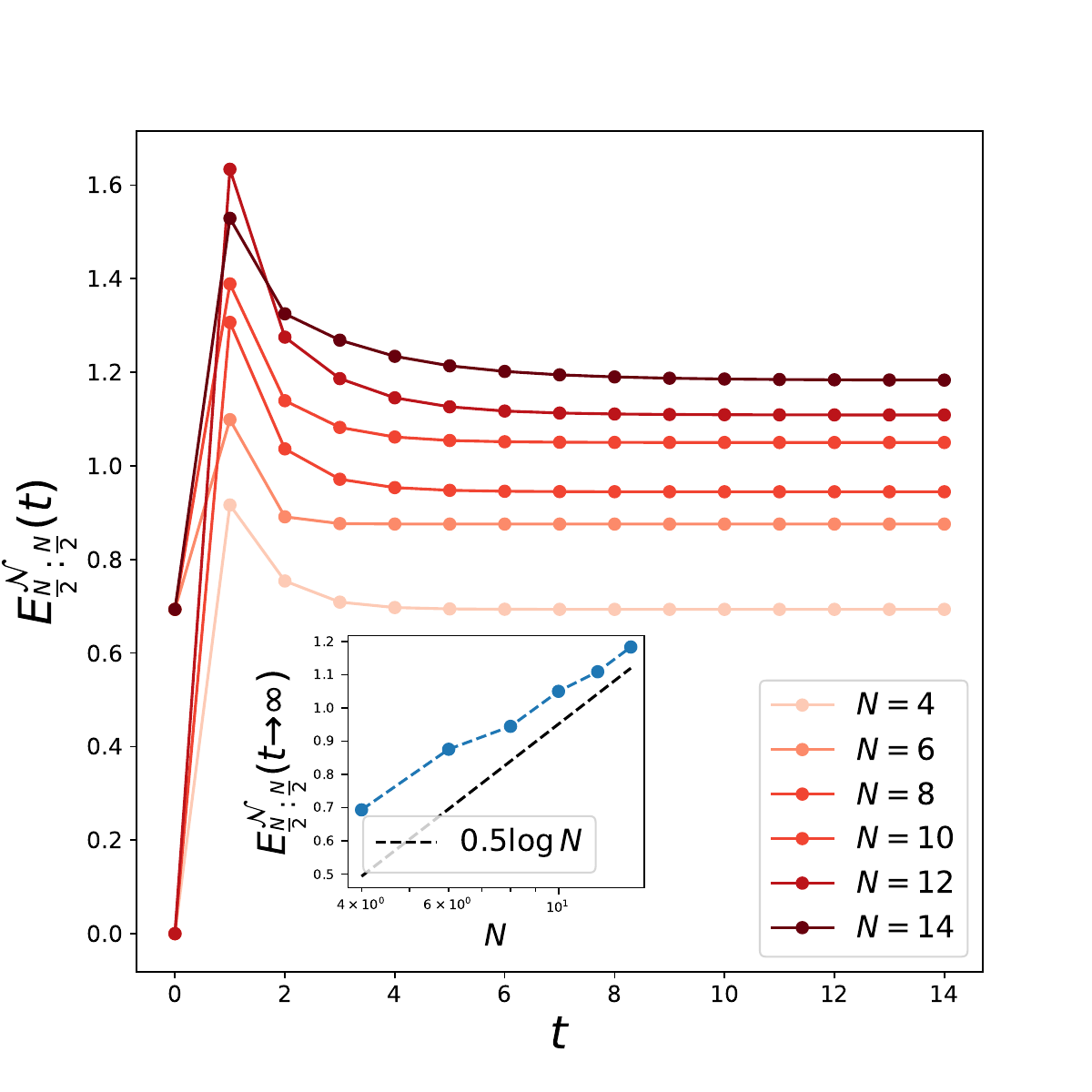}
    \caption{Evolution of half chain negativity for a qubit chain with un-monitored singlet-triplet measurement channel on nearest neighbor qubits, where the initial state is initialized in pairs of nearest singlet. Inset: The steady state half system logarithmic negativity grows logarithmically with system size, as $\sim 0.5 \log N$.} 
    

    \label{fig:intro}
\end{figure}

In Fig.~\ref{fig:intro}, we find that the half chain logarithmic negativity eventually reaches a steady state value with increasing number of measurement layers, and settles to a value which increases as the system size is increased. 


The exact diagonalization numerics are limited to small system sizes; however, we will later rigorously prove that the asymptotic scaling of the bipartite entanglement of formation~\cite{bennett_mixed_1996}, and the  logarithmic negativity for the half chain bipartition, scales as $0.5\log{N}$ in the thermodynamic limit. In this example, the initial state is in a specific spin sector of $SU(2)$, namely the invariant sector, $\{ \ket{\psi} \mid U_{g}\ket{\psi} = \ket{\psi} \}$. The Kraus operators of the channel, $P_{s,i}$ and $P_{t,i}$, also commute with $U_{g}$ for all $g \in SU(2)$. This is defined as strong symmetry of quantum channels~\cite{Buca_2012, Albert_2014,de_Groot_2022}, as it preserves the strong symmetry of the states~\cite{de_Groot_2022}. Building on these ideas, we are able to prove several general results about the entanglement structure of the steady states of such strongly symmetric quantum channels. The results do not depend on the channel being brickwork; any unital symmetric channel that produces the maximally mixed state within a particular symmetry sector will lead to the same behavior of entanglement, which we summarize next.



\subsection{Summary of main results}
\label{sec:main_results}

First, we characterise the sufficient conditions when there is a unique stationary state for strongly symmetric channels.

\begin{result}\label{res:SS} (Sec.~\ref{sec:channel_thm}) For a strongly symmetric channel which is algebraically complete, there is a unique stationary state in every diagonal symmetry sector of the mixed states. Furthermore, if the channel is unital, the maximally mixed state in the symmetry sector is the unique steady state.\end{result}

By algebraically complete, we mean that the Kraus operators corresponding to the channel and their products and linear combinations can generate all symmetric operators. Diagonal sector refers to those mixed state density matrices whose eigenvectors are also eigenvectors of the complete set of commuting operators (CSCO) of the symmetry group. Unital channels refer to channels which preserve the identity matrix, such as projective measurement channels.


Next, we characterise the entanglement of such maximally mixed symmetric states. Two operationally meaningful notions of bipartite entanglement in mixed states are: entanglement of formation~\cite{bennett_mixed_1996} (related to the minimum number of Bell pairs necessary to prepare the state by local operations and classical communication, i.e. LOCC), and the entanglement of distillation~\cite{bennett_purification_1996} (the maximum number of Bell pairs that can be distilled from the state by LOCC, in the asymptotic limit of number of samples)~\footnote{For pure states, these two measures are equal to the entanglement entropy~\cite{Bennett_1996Concentrating}.}. The entanglement of formation is generally a hard quantity to compute even for few-body states.

Logarithmic negativity~\cite{vidal_computable_2002} is an easily computable measure of entanglement in mixed states, and an upper-bound to the distillable entanglement. Logarithmic negativity measures the violation of the positive partial transpose (PPT) condition - a necessary but generally not sufficient condition for separable states~\cite{Peres_1996, Horodecki_1996, Horodecki_1998MixedState}. 
Note, on the other hand, that the entanglement of formation is a faithful measure of quantum entanglement, and hence it being zero/ non-zero implies that the state is necessarily separable/ entangled.


One central result of our work is that the entanglement of formation can be exactly computed for steady states of strongly symmetric channels in the invariant sector, and furthermore, that there is a symmetry-based classification of the amount of entanglement in such states.

We generalize the maximally mixed state in the spin $0$ sector for the $SU(2)$ spin chain to any symmetry, dubbed the maximally mixed invariant state (MMIS), referring to the equal mixture of all pure states in the invariant (equivalently, trivial or singlet) sector of the symmetry group. To be concrete, we are considering such states in the tensor product Hilbert space of $N$ qudits, where the local $d$ dimensional qudits carries an irreducible representation of the symmetry group. We are interested in studying the entanglement property of such mixed states not only due to their ubiquity as steady states of symmetric channels, but also because they reveal intrinsic properties of the symmetry action, as this is the only information they contain. We show the following result:

\begin{result}\label{res:Ent} (Sec.~\ref{sec:entanglement_equality}) The  maximally mixed invariant state has the property that for any real space bipartition, the entanglement of distillation  and the entanglement of formation  are equal.\end{result}

The equality of the entanglement of formation and the entanglement of distillation follows from the fact that such mixed states are ensembles of pure states that are distinguishable by LOCC, without destroying quantum entanglement~\cite{horodecki_entanglement_1998,livine_entanglement_2005}.


Furthermore, this result also provides an optimal decomposition of the mixed state that determines the entanglement of formation, as well as the optimal strategy for distilling Bell pairs from the mixed state. This remarkable property allow us to derive a closed form expression for the bipartite entanglement for the maximally mixed invariant state. 

Using the expressions for bipartite entanglement measure, we obtain the following classification of the entanglement structure for different symmetries,

\begin{result}\label{res:EBC} The maximally mixed invariant state (MMIS) has the following entanglement based classification (Sec.~\ref{sec:entanglement_classification}):
\begin{itemize}
    \item{The MMIS  for any Abelian (discrete or continuous) symmetries, such as $\mathbb{Z}_{2}$ or $U(1)$, is fully separable.}
    \item{The MMIS for any non-Abelian symmetries is entangled, i.e. non-separable.\footnote{For certain symmetry groups, the MMIS might be separable for some fine-tuned bipartitions $A:B$. See Sec.~\ref{sec:entanglement_classification} for more details.} Furthermore,}
    \begin{itemize}
        \item{for finite discrete\footnote{All discrete groups treated here are finite.} non-Abelian symmetries, such as the permutation group $S_{3}$, the bipartite entanglement for any bipartition is non-zero, but finite $\sim O(1)$ in the thermodynamic limit, $N\to \infty$ (Sec.~\ref{sec:s3_entanglement}).}
        \item{for symmetries described by compact semisimple Lie groups, such as $SU(M)$ or $SO(M)$, the bipartite entanglement of half of the system is $\sim\log N$ in the thermodynamic limit $N\to \infty$. More precisely, we provide a general formula for half system entanglement of formation for any such groups in terms of the dimension ($\dim \mathfrak{g}$) and rank ($\dim \mathfrak{h}$) of its Lie algebra, $\frac{\dim \mathfrak{g}-\dim \mathfrak{h}}{4}\log N + O(1)$ (Sec.~\ref{sec:lie_mmis}).}
    \end{itemize}
\end{itemize}
\end{result}

All the results are independent of spatial dimension, since all MMISs are invariant under permutations of the sites (See Lemma \ref{lemma:permutation_invariance}). 

These results show that the steady state that we obtained in the numerical experiment in Fig.~\ref{fig:intro} for the $SU(2)$ symmetric channel is indeed a genuinely long-range entangled mixed state. Here, by long-range entangled (LRE) mixed state, we mean that the state cannot be decomposed into a classical mixture of states which are all short-range entangled, i.e. SRE (product states with finite depth local unitaries acted on them). Furthermore, the distillation protocol in Result~\ref{res:Ent} provides an explicit strategy to utilise this entanglement. This provides an example where the imposition of strong symmetry in an otherwise entropy-increasing unital quantum channel gives rise to a genuinely long-range entangled steady state with logarithmically large bipartite entanglement. We can also compute other probes of entanglement for such mixed states, namely logarithmic negativity and operator entanglement~\cite{cirac_MPDO_first}, the latter being connected to MPO representations. 

While the MMIS for non-Abelian symmetries are globally entangled, the global constraint is not revealed in any local experiment. This intuition is confirmed in the following result,

\begin{result}\label{res:exp_values}
    The $k$-site reduced density matrix of the MMIS converges to the identity state in the thermodynamic limit $N\to \infty$, when $k$ is constant (Sec.~\ref{sec:exp_values}).
\end{result}

Despite the apparent triviality of local reduced density matrices, for finite systems ($N < \infty$), the MMIS of compact semisimple Lie groups are long-range correlated, with algebraically decaying correlations $\sim 1/N$ (Sec.~\ref{sec:LRO}). For discrete non-Abelian symmetries, the MMIS is strictly short-range correlated (Secs.~\ref{sec:exp_values},\ref{sec:s3_entanglement}).

The genuine long-range entanglement in the MMIS of compact semisimple Lie groups leads to the question: how easy is it to prepare such state from pure product states? We have the following result that addresses this:

\begin{result}\label{res:time}
The algebraically decaying correlation in the MMIS of compact semisimple Lie groups implies a lower bound of $\Omega(N)$ depth for preparing such states from pure product states using local channels. The entanglement of formation provides a lower bound $\Omega(\log N)$ of the depth necessary for adaptive circuits (local quantum operations and instanteneous classical communication) to prepare such states (Sec.~\ref{sec:preparation}).
\end{result}

One can further ask if such states are stable mixed state `phases of matter'. We discuss in Sec.~\ref{sec:stability_swssb} that the 
MMIS for any symmetry group displays a novel form of symmetry-breaking pattern that is specific to mixed states, namely, strong-to-weak spontaneous symmetry breaking (SW-SSB)~\cite{lessaStrongtoWeakSpontaneousSymmetry2024, salaSpontaneousStrongSymmetry2024}, which are diagnosed by the non-vanishing of the so-called fidelity correlators at large separations. We have the following result:

\begin{result}\label{res:SWSSB} 
The MMIS of compact semisimple Lie groups have strong-to-weak spontaneous symmetry  breaking (Sec.~\ref{sec:stability_swssb}). 
\end{result}

Using the stability theorems of SW-SSB proved in~\cite{lessaStrongtoWeakSpontaneousSymmetry2024}, this result implies that the MMIS of such groups cannot be converted into a symmetric pure product state using short-depth symmetric local channels. 

The MMIS can be thought of as the infinite temperature state within the trivial symmetry sector. We also explore whether such long-range entanglement with $O(\log N)$ entanglement exists at finite temperatures. Using some heuristic arguments in Sec.~\ref{sec:thermal_mmis}, we conjecture that the infinite temperature MMIS forms a universality class distinct from any finite temperature states.




\subsection{Relation to previous works}




Several pioneering works from the early days of quantum information highlighted how entanglement of formation can be computed for states with enhanced symmetries~\cite{Werner_1989, Vollbrecht_2001}. Moreover, the equality of entanglement of distillation and entanglement of formation was established for states that are composed of ensembles of locally orthogonal pure states in~\cite{horodecki_entanglement_1998}. It was also shown in several works relevant to quantum gravity, that such equality of entanglement measures also apply to the maximally mixed state in the singlet sector of $SU(2)$ invariant spin systems~\cite{livine_entanglement_2005, livine_quantum_2006}. In this work, we have generalized these results to the trivial sector of any general symmetry group, uncovering the symmetry-based classification of the entanglement structure of such states in the presence of non-Abelian symmetries.

Several recent works have investigated the entanglement structure of states with non-Abelian symmetries; however, much of these works have focused on pure states. For pure states, it is known that non-Abelian symmetries can enhance their entanglement~\cite{Majidy_2022, bianchi2024nonabelian}, even precluding the possibility of many-body localization~\cite{Potter_2016}. Measurement circuits preserving a global $SU(2)$ symmetry also leads to logarithmically large bipartite entanglement in their \textit{pure state} trajectories as was reported numerically in~\cite{Majidy_2023}. Our result about the logarithmically scaling entanglement of formation of the maximally mixed state in the singlet sector provides an analytical justification for a logarithmic lower bound on the averaged entanglement of the pure state trajectories: since the entanglement of formation minimizes the average entanglement amongst all pure state decompositions of the state (which include measurement trajectories). Such logarithmic enhancement also appears in \textit{pure} steady states with emergent $SU(2)$ symmetries in certain adaptive circuits~\cite{hauser_continuous_2023}. 

The interplay of \textit{mixed} state entanglement structure and non-Abelian symmetry has been treated on the `weak' regime, in which the non-commuting charges are allowed to be transported locally, and conserved globally with the bath. In that context, the existence of non-commuting charges necessitates new approaches towards establishing quantum thermalization to a Gibbs state (see~\cite{Majidy_2023b} for a recent review). In this work, we instead highlight on how `strong' non-Abelian symmetries (where the system does not exchange the non-commuting charges with the environment) lead to long-range entanglement in \textit{mixed} steady states, and are very different from Gibbs states of local Hamiltonians~\cite{Kuwahara_2020}, which become even unentangled at high enough temperatures \cite{bakshi_high-temperature_2024}.

Recently, highly entangled \textit{mixed} steady states have been found to arise in a different context: in~\cite{li_hilbert_2023} it was found that `strong' kinetic constraints in unital quantum channel leads to highly entangled steady states, generalizing Hilbert space fragmentation to open systems. Our work demonstrates that even conventional non-Abelian symmetries can do so. 

Prior work on quantum channels with strong non-Abelian symmetry has focused on the number of steady states~\cite{zhang_stationary_2020}. Sufficient conditions for unique steady states within symmetry sectors have been studied for generic unital quantum channels with strong Abelian symmetry in~\cite{Buca_2012, Yoshida_uniqueness_2024}.  












\section{Symmetries of states and channels}\label{sec:symm_def}

Consider a tensor product space of $N$ degrees of freedom carrying an irreducible representation $\mathcal{V}$ of the group $G$. The total Hilbert space, $\mathcal{H} \equiv \mathcal{V}^{\otimes N}$ decomposes into a direct sum of irreducible representations (irreps),
\begin{align}\label{eq:multiplicity_decomp}
    \mathcal{V}^{\otimes N} = \bigoplus_{J} \mathfrak{C}^{N}_{J} V_{J}.
\end{align}
Here, each irrep $V_{J}$ is labeled by the irrep label $J$. Equivalent irreps $V_{J}$, of dimension $d_J$, appear $\mathfrak{C}^{N}_J$ number of times in the decomposition, refered to as the mutiplicity of the irrep $J$ in the total Hilbert space. An orthonormal basis of $\mathcal{H}$ is labeled as $\ket{J,M,\alpha}$, with $J$ being the irrep label, $M = 1, \cdots, d_{J}$ being the label for a state within an irrep, and $\alpha =  1, \cdots, \mathfrak{C}_{J}^N$ being the multiplicity label. For general non-Abelian symmetries $d_{J} \geq 1$, while for Abelian symmetries, $d_{J} = 1$. 





Any unitary representation of the group element $g \in G$ in $\mathcal{H}$ can be decomposed into its irreps, $U(g) = \bigoplus_{J}U^{J}(g)$, and acts trivially in the space of multiplicities, $U^{J}(g) = \sum_{M,M' = 1}^{d_{J}}\sum_{\alpha, \beta = 1}^{\mathfrak{C}^{N}_{J}}\left[U^{J}(g)\right]_{M,M'}\delta_{\alpha,\beta}\ket{J,M,\alpha}\bra{J,M',\beta}$. A symmetric operator $O$ is defined such that it commutes with (representations of) all group elements, $[O, U(g)] = 0, \: \forall g \in G$. By Schur's lemma~\cite{fulton_representation_2004}, this implies that $O$ acts as identity within each irrep, but non-trivially on the multiplicity space,
\begin{gather}
O = \bigoplus_{J}O^{J}, \text{ with} \nonumber \\ 
O^{J} = \sum_{M,M',\alpha,\beta}\left[O^{j}\right]_{\alpha,\beta}\delta_{M,M'}\ket{J,M,\alpha}\bra{J,M',\beta}. \label{eq:symm_op}
\end{gather} 

\subsection{Symmetric mixed states}

A density matrix $\rho$ can be symmetric under $G$ in two inequivalent ways, namely weak and strong. $\rho$ is weakly symmetric if $U(g) \rho U^{\dagger}(g) = \rho$ for all $g\in G$. 

In this work, we will focus on strongly symmetric mixed states, which are defined as follows:

\begin{definition}$\rho$ is strongly symmetric if $U(g)\rho = \lambda(g)\rho$, for a $U(1)$ valued representation $\lambda(g)$ of $G$. 
\end{definition}
Mixed states with specific symmetry charges $J$ (irrep labels) belong to the vector space of linear operators on the decomposed Hilbert space $\mathcal{H}^{J}$. We can further consider a subspace of $\mathcal{H}^{J}$ with well defined $M$ label spanned by the states $\ket{J,M,\alpha}$ for different $\alpha$. We call this $(J,M)$ sector of the Hilbert space $\mathcal{H}^{JM}$. Analogously, a linear operator in the $(J,M)$ sector, denoted by $\mathcal{B}^{JM}$, is spanned by the orthogonal basis of $\mathfrak{C}_{J}^{2}$ operators $\ket{J,M,\alpha}\bra{J,M,\beta}$.

\subsection{Strongly symmetric quantum channels}

We are interested in characterising  quantum channels $\mathfrak{L}$ that are specified by the Kraus operators $\{K_{a}\}$ (satisfying $\sum_{a}K_{a}^{\dagger}K_{a} = \mathbb{I}$) and acts on density matrices as $\mathfrak{L}(\rho) = \sum_{a}K_{a}\rho K_{a}^{\dagger}$. We first define the notion of strong symmetry for channels: 

\begin{definition} 
\textbf{Strongly symmetric channels} are those channels whose Kraus operators are symmetric, i.e. $[K_{a}, U(g)] = 0$ for all $g\in G$, and thus can be decomposed as Eq.~\eqref{eq:symm_op}. 
\end{definition}

\begin{definition}We  define (algebraically) \textbf{complete strongly symmetric} Kraus operators to be a set of Kraus operators $\{\left[K^{J}_{a}\right]\}$ that can be linearly combined or multiplied to generate any symmetric operator $O^{J}$.
\end{definition}

The adjoint channel $\mathfrak{L}^{\dagger}$ acts on operators as $\mathfrak{L}^{\dagger}(O) = \sum_{a}K_{a}^{\dagger}O K_{a}$. For notation purposes, we introduce the identity channel, $\mathfrak{I}(\rho) = \rho$.  It is worth recalling that if $\forall a ~ [O,K_{a}] = 0$, then $\left(\mathfrak{L}^{\dagger} - \mathfrak{I}\right)(O) = 0$.

Consider an initial quantum state in the $(J,M)$ sector. The channel being strongly symmetric will preserve both $J$ and $M$ labels, while possibly mixing states from different multiplicities. We are interested in exploring the existence and uniqueness of steady states of symmetric channels in specific symmetry sectors.

In this work, we will characterise correlations and entanglement structure of mixed steady states in the $(J,M)$ sector. We focus on the most `featureless' mixed state in a symmetry sector, namely the Identity matrix within the $(J,M)$ sector, $\rho_{JM} = \frac{1}{\mathfrak{C}_{J}}\sum_{\alpha}\ket{JM\alpha}\bra{JM\alpha}$. We will also sometimes consider the maximally mixed states of all states in the irrep $J$ subspace, $\rho_{J} = \frac{1}{d_J\mathfrak{C}_{J}}\sum_{M,\alpha}\ket{JM\alpha}\bra{JM\alpha}$. 


Note that $\rho_{JM}$ is a strongly symmetric mixed \emph{invariant} state only for the invariant irrep, $J = \triv$. Since $d_J = 1$, the $M$ label is unnecessary.

Next, we show that $\rho_{JM}$ is the unique steady state of strongly symmetric channels that are algebraically complete and unital.


 
\subsection{Unique steady states of complete strongly symmetric unital channels}\label{sec:channel_thm}
We have the following results for steady states of strongly symmetric channels:

\begin{theorem}
There is at least one steady state in $\mathcal{B}^{JM}$ for any strongly symmetric  channel.
\end{theorem}

This proof follows an argument for the lower bound on the degeneracies of the steady states of non-Abelian symmetric Lindbladians in ~\cite{zhang_stationary_2020}.

\begin{proof}
    We construct an operator in $\mathcal{B}^{JM}$, $\overline{O} = \sum_{\alpha}\ket{JM\alpha}\bra{JM\alpha}$. We can check that $\forall a ~ [K_{a}, \overline{O}] = 0$. This proves that $\overline{O} \in \text{ker}_{\mathcal{B}^{JM}}\left(\mathfrak{L}^{\dagger} - \mathfrak{I}\right)$, where the RHS refers to the kernel of $\mathfrak{L}^{\dagger} - \mathfrak{I}$ in the subspace $\mathcal{B}^{JM}$. Since the dimensions of the kernels of a linear operator and its adjoint are the same, this implies that there exists at least one state $\rho^{*} \in \left(\mathfrak{L}- \mathfrak{I}\right)$ which is  a steady state of the channel $\mathfrak{L}$.
\end{proof}

\begin{theorem}
If the quantum channel is \textit{complete and strongly symmetric}, then there is a unique steady state in $\mathcal{B}^{JM}$.
\end{theorem}
This can be proved by generalizing the proof presented in ~\cite{Yoshida_uniqueness_2024} for a similar result for general Lindbladians with Abelian strong symmetries. We generalize the result for quantum channels symmetric under any group, including non-Abelian ones.
\begin{proof}
First, we show that any (unnormalized) steady state $\rho$ either has no zero eigenvalues or is the zero operator. If $\rho$ has a zero eigenvalue, there exists a state $\ket{\psi} \in \mathcal{B}^{JM}$ such that $0 = \bra{\psi}\rho\ket{\psi} = \sum_{a}\bra{\psi}K_{a}\rho K_{a}^{\dagger}\ket{\psi}$. This implies $\forall a ~ \sqrt{\rho}K_{a}^{\dagger}\ket{\psi} = 0$, or equivalently, $K_{a}^{\dagger}\ket{\psi} \in \text{ker}\rho$. By the assumption of completeness of the Kraus operators, $\{K_{a}^{\dagger}\ket{\psi}\}$ spans $\mathcal{B}^{JM}$, thus $\text{ker}\rho  = \mathcal{H}^{JM}$. Since $\rho$ is supported in $\mathcal{B}^{JM}$, then $\rho = 0$.

Next, following~\cite{Yoshida_uniqueness_2024}, we assume $\rho_{1}$ and $\rho_{2}$ are two distinct steady states in $\mathcal{B}^{JM}$. We construct a steady ``state'' $\rho(x) = (1-x)\rho_{1}-x\rho_{2}$ for $0\leq x \leq 1$.  All the eigenvalues of $\rho(0)$ are positive while all eigenvalues of $\rho(1)$ are negative. By the continuity of the spectrum of $\rho(x)$, there exists a real number $0\leq x_{0} \leq 1$ such that
the minimum eigenvalue of $\rho(x_{0})$ is zero, and by previous result, such a $\rho(x_{0})$ must be zero. By taking trace of $\rho(x_{0})$, we get that $x_{0} = 1/2$ and thus $\rho_{1} = \rho_{2}$ which contradicts our assumption of distinctness of the two steady states.
\end{proof}

\begin{lemma}
If the quantum channel is complete, symmetric, and \textit{unital}, the steady state is the projected identity matrix,
$\rho_{JM} = \frac{1}{\mathfrak{C}_{J}}\sum_{\alpha}\ket{JM\alpha}\bra{JM\alpha}$.
\end{lemma}
\begin{proof}
    The unitality condition is $\sum_{a}K_{a}K_{a}^{\dagger} = \mathbb{I}$. Using the decomposition of symmetric Kraus operators as in Eq.~\eqref{eq:symm_op}, we can explicitly check that $\rho_{JM}$ is a steady state. By the completeness of Kraus operators, this is the unique steady state.
\end{proof}



\section{Maximally mixed invariant state }\label{sec:mms_correlations}

Consider the maximally mixed invariant state (MMIS), which, as we described in earlier sections, is the unique steady state in the invariant sector for any strongly symmetric, complete, and unital channels. Let $\ket{\alpha}$ be an orthonormal basis of invariant states labeled by the multiplicity,
\begin{align}\label{eq:mms_invariant}
    \rho_{\triv}^{N} = \frac{1}{\mathfrak{C}_{\triv}^{N}}\sum_{\alpha = 1}^{\mathfrak{C}_{\triv}^{N}} \ket{\alpha}\bra{\alpha}  = P_{\triv} \frac{\mathbb{I}}{d_{\triv}},
\end{align}
where $P_{\triv}$ is the projector on the trivial sector, and $\mathbb{I}$ is the identity operator on the whole Hilbert space.

We study the entanglement structure and correlations in the MMIS. Naively, one would expect the MMIS to be very simple and uncorrelated, since it is nothing but the identity matrix (which is a separable state), restricted to a specific symmetry sector. However, the projector onto the trivial sector, $P_{\triv}$, imposes stringent global constraints on the correlation structure of the state. In particular, we will show that for non-Abelian symmetries, this state is necessarily entangled (except for caveat discussed in Sec. \ref{sec:entanglement_classification})\footnote{\label{footnote:Abelian_MMIS}The MMIS for any Abelian group is always fully separable (not entangled), which can be argued as follows. The global charge constraint $g \rho = \lambda(g) \rho$ imposes a global classical constraint, and due to the Abelian nature of the symmetry, one can decompose the global charge $\lambda : G \to U(1)$ into local charges $\lambda_i$, with $\prod_i \lambda_i = \lambda$. Thus, the MMIS is the mixture of all pure product states with charges summing up to the global charge, $\rho^{N}_{\triv} = \sum_{\prod_i \lambda_i = \lambda} \ketbra{\vec{\lambda}}{\vec{\lambda}}$.}.

Before delving into the entanglement structure of the MMIS, we study the correlation functions of the MMIS, which already hints at their non-triviality. As representatives of discrete and continuous non-Abelian symmetries, we will refer to the example of MMIS of $S_3$ and $SU(2)$ symmetries on a spin chain repeatedly. The results for the $S_3$ MMIS are shown in a later section, Sec.~\ref{sec:s3_entanglement} and Appendix~\ref{appsec:s3_details}, whereas the ones for $SU(2)$ are shown in Sec. \ref{sec:lie_mmis_su2}. 

One result we will use to compute correlations in the MMIS is that the maximally mixed state within a symmetry sector of an on-site symmetry is invariant under permutation of its local degrees of freedom:

\begin{lemma}\label{lemma:permutation_invariance}
    The maximally mixed state $\rho_J$ in a symmetry sector (irrep) $J$ of an on-site symmetry $U(g) = \otimes_{i=1}^N U_i(g)$ is invariant under permutation of the sites. In other words, it has weak $S_N$ symmetry.
\end{lemma}
\begin{proof}
    Since $U$ is on-site, it commutes with the permutation group action $\pi$: $[U(g), \pi(\sigma)] = 0$ for all $g \in G$ and $\sigma \in S_N$. Thus, from Schur's lemma, $\pi(\sigma)$ can only have non-zero coefficients between different multiplicity sectors of the same irrep, $\pi = \oplus_J \pi_J$. Since $\rho_J$ acts as a multiple of the identity matrix in the multiplicity space of $J$ and zero elsewhere, then $[\rho_J, \pi(\sigma)] \propto [P_J, \pi_J(\sigma)] = 0$. 
\end{proof}
This applies to MMIS by taking $J = \triv$. This means that any notion of geometrical locality, such as spatial dimension, is absent for the $\rho_J$ states.

\subsection{Long-range correlations in MMIS} \label{sec:LRO}
First, we demonstrate that the MMIS of 
semisimple non-Abelian Lie groups such as
$SU(2)$, which we described in the warm-up example in the introduction (Fig.~\ref{fig:intro}), cannot be expressed as a mixture of short-range entangled pure states. To do this, we first prove that these symmetric states are long-range ordered, and use an argument noted in~\cite{lu_mixed-state_2023} to prove the long-range nature of the entanglement of such mixed states.

For generality, consider semisimple non-Abelian Lie groups. Let $(X_\alpha)_{\alpha=1}^{\dim \mathfrak{g}}$ be an orthonormal basis of the Lie algebra $\mathfrak{g}$ of $G$ with respect to the Killing form. Then, we can define the Casimir element by $C = \sum_\alpha X_\alpha X^\alpha$, where $X^\alpha$ is an element of the dual basis. In the case of the on-site representation $U = \otimes_{i=1}^N U_i$, the representation of the global Casimir element is $U(C) = -\sum_{\alpha} \left(\sum_i U(X_{i,\alpha}) \right)^2$. It can be proven that $U(C)$ for a reducible representation $U$ is constant in each irrep $J$, equal to $c_J$, such that $U(C) = \sum_J c_J P_J$. Moreover, $c_J = 0$ if, and only if, $J = \triv$ \cite{hall_lie_2015}. With this in mind, we can calculate
\begin{align}\label{eq:casimir_correlator}
      0 & = U(C) \rho_\triv^N = - \sum_\alpha \sum_{i,j=1}^N U(X_{i, \alpha}) U(X_{j, \alpha}) \rho_\triv^N \nn \\
      & = \sum_{i=1}^N U_i(C_i) \rho_\triv^N - \sum_\alpha \sum_{i \neq j}^N U(X_{i,\alpha}) U(X_{j, \alpha}) \rho_\triv^N.
\end{align}
Using the fact that $\rho_\triv^N$ is invariant under permutations (See Lemma \ref{lemma:permutation_invariance}), then taking the trace of the equation above gives 
\begin{equation}\label{eq:long-range_correlations}
    - \Tr\left[\rho_\triv^N \sum_\alpha U(X_{x,\alpha}) U(X_{y, \alpha}) \right] = -\frac{c_\mathcal{V}}{N-1},
\end{equation}
for all sites $x \neq y$, where $\mathcal{V}$ is the on-site Hilbert space, assumed to be an irrep of the on-site symmetry. If $\dim \mathcal{V} > 1$, then $c_\mathcal{V} > 0$ \cite{hall_lie_2015}, which implies $\rho^N_{\triv}$ has long-range correlations that vanish in the thermodynamic limit. In the case of $SU(2)$ acting on a spin-$1/2$ chain, for example, Eq.~\eqref{eq:long-range_correlations} becomes the usual spin-spin correlation function,
\begin{equation}\label{eq:su2_correlation}
    \Tr[\rho_\triv^N \vb{S}_i \cdot \vb{S}_j] = -\frac{3}{4(N-1)}.    
\end{equation}




By definition of the MMIS, the expectation values of the single site Lie algebra generator vanishes, $\Tr[\rho_\triv^{N} U(X_{x,\alpha})] =0$. Thus, the connected correlation function also scales as $1/N$ for any arbitrary locations $i,j$ (for example, the spin-spin connected correlation of the $SU(2)$ MMIS is also $-3/4(N-1)$ by Eq.~\eqref{eq:su2_correlation}). Using this result, we can prove that MMIS is necessarily long-range entangled, using an argument from~\cite{lu_mixed-state_2023}. By the Lieb-Robinson bound, the connected correlation function for any short-range entangled pure state (any state obtainable from product state using a finite depth unitary circuit) must decay exponentially with $|i-j|$ at large distances~\cite{Bravyi_2006}. This contradicts the algebraic decay of correlations in Eq.~\eqref{eq:long-range_correlations} by taking $|i-j| \sim N$.


For discrete non-Abelian groups such as $S_3$, the connected correlation function between any two qudits in the MMIS decays exponentially with the system size (discussed below, in Sec.~\ref{sec:exp_values}) which hints at the short-range entangled nature of the state. Particularly, we study the case of $S_3$ group in more detail in   Sec.~\ref{sec:s3_entanglement} and Appendix~\ref{appsec:s3_details} by explicit calculations.

\subsection{Local indistinguishability from identity state}\label{sec:exp_values}


In Section \ref{sec:LRO}, we showed that continuous non-Abelian Lie groups have MMIS with $1/N$-decaying correlations for certain two-body operators. Although this is sufficient to prove long-range order of the MMIS (the correlation amplitude does not depend on the distance between the two sites and it decays slower than exponential), it still goes to zero in the thermodynamic limit $N \to \infty$. This raises the question of whether the same happens for all expectation values of $k$-body observables $A_k$. Intuitively, since the MMIS is constrained only by the global strong symmetry, any local experiment should reveal no information in the thermodynamic limit, and thus the MMIS should be locally indistinguishable from the identity state. Thinking in terms of thermal states, this is related to the equivalence of canonical and grand canonical ensembles under local observables at infinite temperature. Here, we confirm this expectation in the following sense: 

        

\begin{theorem}\label{thm:local_indist}
    The $k$-site reduced density matrix $\rho^{N,k}_\triv \equiv \Tr_{\{1,\ldots,k\}^c}[\rho^N_\triv]$ of the MMIS under a faithful representation of $G$ converges to the maximally mixed state $\propto \mathbb{I}_{\mathcal{V}^{\otimes k}}$ with
    \begin{itemize}
        \item exponentially small error $O(e^{-\alpha (N-k)})$ in trace distance for finite non-Abelian groups, and
        \item algebraically small error $O(k / N)$ in trace distance for semisimple Lie groups. \footnote{Note that for any symmetry group $G$, the 1-site reduced density matrix ($k =1$) is always exactly equal to the identity state $\mathbb{I}_{\mathcal{V}}/d_{\mathcal{V}}$. The reason is that the summation in Eq.~\eqref{eq:reduced_MMIS} has only one allowed term, associated with the fixed irrep of a single site.}
    \end{itemize}
\end{theorem}

An immediate corollary is that the expectation value of any $k$-body operator $A_k$ in the MMIS $\rho^N_\triv$ is
    \begin{equation}\label{eq:exp_vals_corollary}
        \langle A_k \rangle = \frac{1}{2^k} \Tr[A_k] + %
        \begin{cases}
            O(\norm{A_k} e^{-\alpha(N-k)}), & \text{if $G$ is finite,} \\
            O(\norm{A_k} k / N), & \text{if $G$ is a Lie group.} \\ 
        \end{cases}
    \end{equation}
    

In appendix \ref{appsec:exp_values_proof} we provide a proof sketch of these results. In the next section we show that, globally, however, the state is highly entangled for large bipartitions.

\section{Entanglement structure of MMIS}\label{sec:mms_entanglement}
In the previous section, we showed that the MMIS of $SU(2)$, and in general, semisimple non-Abelian Lie groups, must necessarily be long-range entangled. Here we show that we can describe their entanglement structure even more precisely.

\subsection{Mixed state entanglement measures}\label{sec:entanglement_measures}

In order to diagnose entanglement and quantum correlations, we review some measures of entanglement in mixed states~\cite{horodecki_quantum_2009}. A mixed state $\rho$ can always be expressed (non-uniquely) as a convex mixture of pure states from an orthonormal basis set, $\rho = \sum_i p_i \ket{\psi_i}\bra{\psi_i}$. For a given bipartition, $\rho$ is separable if it can be expressed as a mixture of product states with respect to that bipartition; otherwise the state is entangled. 

The entanglement of formation $E^f_{A:B}(\rho)$~\cite{bennett_mixed_1996} for a particular bipartition $A:B$ is defined as the average entanglement entropy of these pure states minimized over all possible decompositions, $E^f_{A:B} = \inf_{\{p_{i}, \ket{\psi_i}\}} \sum_i p_{i}S_{A:B}(\ket{\psi_i})$, where $S_{A:B}(\ket{\psi})$ refers to the von Neumann entanglement of the pure state $\psi$ for the bipartition $A:B$. A regularized version of the entanglement of formation~\cite{hayden_asymptotic_2001} (also dubbed `entanglement cost') is defined in the asymptotic limit of infinite copies of the state as $E^F_{A:B}(\rho) = \lim_{n\to \infty}E^f_{A:B}(\rho^{\otimes n})/n$. Due to the non-additivity of entanglement of formation, $E^F_{A:B}(\rho)\leq E^f_{A:B}(\rho)$. 

Another measure of mixed state entanglement is the entanglement of distillation $E^D_{A:B}$~\cite{bennett_purification_1996}. It  is defined as the maximum number of Bell pairs Alice and Bob (localised on $A$ and $B$ respectively) can obtain via local operations and classical communication (LOCC), per sample of the state $\rho$, asymptotically, when they have access to an infinite samples of $\rho$. Formally, we can write $E^D_{A:B}(\rho) = \sup_{s \in \mathcal{S}}  \lim_{n \to \infty} N_{\text{Bell}}(\rho^{\otimes n}|s)/n$, where $N_{\text{Bell}}(\rho|s)$ is the number of Bell pairs Alice and Bob can extract form state $\rho$ via the LOCC process (equivalently distillation strategy) $s$, and $\mathcal{S}$ is the set of all distillation strategies.

It is known that \textit{good} bipartite entanglement measures $E_{A:B}$ that satisfy several desirable properties (monotonicity, additivity, convexity, continuity)  are bounded by $E^D_{A:B}$ and $E^F_{A:B}$~\cite{hayden_asymptotic_2001}, as 
\begin{align}\label{eq:ent_hierarchy}
    E^{D}_{A:B} \leq E_{A:B} \leq E^{F}_{A:B}\leq E^f_{A:B}.
\end{align} 
One such \textit{good} bipartite entanglement measure that is geometric in nature is the relative entropy of entanglement, defined as $E^R_{A:A^c}(\rho) \equiv \min_{\sigma \in \text{SEP}_{A:A^c}} S(\rho \| \sigma)$ \cite{vedral_entanglement_1998, christandl_structure_2006}. It measures the distance from $\rho$ to the set of separable states $\text{SEP}_{A:A^c}$ with respect to the $A:A^c$ bipartition.

Note, some easily computable measures of mixed state entanglement do not satisfy all the properties of entanglement measures listed, and thus do not satisfy this bound. For example, logarithmic entanglement negativity~\cite{vidal_computable_2002}, while an upper bound for $E^D$, can be larger than entanglement of formation. Furthermore, logarithmic negativity, while computable, is not a faithful measure of entanglement: zero logarithmic negativity does not necessarily imply separability.

\subsection{Bipartite decomposition and invariant sector}

We are interested in the bipartite entanglement structure of mixed states defined on the Hilbert space defined in Eq.~\eqref{eq:multiplicity_decomp}. Consider the bipartition of the Hilbert space into $A$ and its complement $A^c$ (where $A$ and $A^c$ consist of $N_A$ and $N_{A^c}$ local degrees of freedom respectively), $\mathcal{H} = \mathcal{H}_{N_{A}}\otimes \mathcal{H}_{N_{A^c}}$. Using the definition of multiplicities in the tensor product space used in Eq.~\eqref{eq:multiplicity_decomp}, we get,
\begin{align}
    \mathcal{V}^{\otimes N} & = \mathcal{V}^{\otimes N_{A}}\otimes \mathcal{V}^{\otimes N_{A^c}} \nonumber\\&= \left(\bigoplus_{J_{A}}\mathfrak{C}^{N_{A}}_{J_{A}}V_{J_{A}}\right)\otimes\left(\bigoplus_{J_{A^c}}\mathfrak{C}^{N_{A^c}}_{J_{A^c}}V_{J_{A^c}}\right) \nonumber\\&= \bigoplus_{J_{A}J_{A^c}}\mathfrak{C}^{N_{A}}_{J_{A}}\mathfrak{C}^{N_{A^c}}_{J_{A^c}}\left(V_{J_{A}}\otimes V_{J_{A^c}}\right)
\end{align}
The tensor product of irreps can be further decomposed into a Clebsch-Gordan series, $V_{J_{A}}\otimes V_{J_{A^c}} = \bigoplus_{J^{\prime}}\mathfrak{D}^{J^{\prime}}_{J_{A}J_{A^c}} V_{J^{\prime}}$. By matching the multiplicities in the bipartite decomposition, we have the following relation, $\mathfrak{C}^{N}_J = \sum_{J_{A},J_{A^c}}^{J}\mathfrak{C}_{J_{A}}^{N_A}\mathfrak{C}_{J_{A^c}}^{N_{A^c}}\mathfrak{D}_{J_{A}J_{A^c}}^{J}$.

Note, the Clebsch-Gordan multiplicities, $\mathfrak{D}_{J_A J_{A^c}}^J$ are non-negative integers, and may in general be greater than 1. For $SU(2)$, one can show that $\mathfrak{D}_{J_{A}J_{A^c}}^{J} = 0 \text{ or } 1$~\cite{hall_lie_2015}. For any general (compact) group we can prove that the Clebsch-Gordan multiplicities corresponding to the invariant (trivial) irrep for any finite (or compact Lie) group is always $0$ or $1$. Character of a group representation is a class function on the group that associates to each group element the trace of the corresponding matrix. The character of the invariant irrep is $\chi_{V_{\triv}} = 1$, by definition. The following result is a simple consequence of Schur orthogonality relations for the irreducible characters of the group~\cite{fulton_representation_2004},
\begin{align}\label{eq:cg_invariant}
    \mathfrak{D}_{J_{A}J_{A^c}}^{\triv} & = \langle \chi_{V_{\triv}}, \chi_{V_{J_{A}}\otimes V_{J_{A^c}}}\rangle \nonumber \\ &=  \int dg \chi_{V_{J_{A}}}(g)\chi_{V_{J_{A^c}}}(g) = 
    \begin{cases}
		1, & \text{if $V_{J_{A}^{*}}\cong V_{J_{A^c}}$}\\
        0, & \text{otherwise.}
    \end{cases}
\end{align}
Note, $J^{*}$ refers to the irrep conjugate to $J$. The formula above refers to a compact Lie group, with the measure of the integral given by the Haar invariant measure. For irreps of a finite group $G$, an analogous formula applies with the integral replaced by a sum $1/|G|\sum_{g\in G}$. 

Hence, the invariant irrep $\triv$ appears in the Clebsch-Gordan decomposition of $V_{J_{A}}\otimes V_{J_{A^c}}$ only when $J_{A^c} = J^{*}_{A}$. Furthermore, in $V_{J_{A}}\otimes V_{J^{*}_{A}}$, $V_{\triv}$ appears with multiplicity $1$. Let us introduce the notation $\mathcal{J}_{AA^c}$ to refer to the set of irreps $J$ which have the following properties: $J$ can be represented in $\mathcal{H}_{A}$ and $J^{*}$ can be represented in $\mathcal{H}_{A^c}$. The result Eq.~\eqref{eq:cg_invariant} implies a general formula for bipartite decomposition of multiplicities for the invariant sector, 
\begin{align}\label{eq:multiplicity_factorization}
    \mathfrak{C}^{N}_{\triv} = \sum_{J \in \mathcal{J}_{AA^c}}\mathfrak{C}_{J}^{N_A}\mathfrak{C}_{J^{*}}^{N_{A^c}}.
\end{align}


We now consider the bipartite decomposition of the maximally mixed invariant state as defined in Eq.~\eqref{eq:mms_invariant}. Eq.~\eqref{eq:multiplicity_factorization} leads to a different convenient description of the multiplicity states $\ket{\alpha}$. For a given bipartition as discussed above, we can consider a global invariant state by constructing a singlet $\ket{JJ^{*}}_{AA^c} \in V_J \otimes V_{J^*}$ out of the irrep $V_{J}$ on $A$ and the irrep $V_{J^{*}}$ on $A^c$. $V_{J}$ appears in  $\mathcal{H}_{A}$ with multiplicity $\mathfrak{C}_{J}^{N_A}$, $V_{J^{*}}$ appears in $\mathcal{H}_{A^c}$ with multiplicity $\mathfrak{C}_{J^*}^{N_{A^c}}$. For each pair of irreps, we can construct an orthonormal singlet state, labeled as, $\ket{JJ^{*},a,b}_{AA^c}$ for $a = 1,\cdots,\mathfrak{C}_{J}^{N_A}$ and $b = 1,\cdots,\mathfrak{C}_{J^{*}}^{N_{A^c}}$. Crucially, such states form a complete orthonormal basis for all singlet states realised in $\mathcal{H}$, with the completeness guaranteed by Eq.~\eqref{eq:multiplicity_factorization}. Thus Eq.~\eqref{eq:mms_invariant} can be recast as,
\begin{align}\label{eq:mms_decomp}
    \rho_{\triv}^{N} = \frac{1}{\mathfrak{C}_{\triv}^{N}}\sum_{J \in \mathcal{J}_{AA^c}}\sum_{a = 1}^{\mathfrak{C}_{J}^{N_{A}}}\sum_{b = 1}^{\mathfrak{C}_{J^*}^{N_{A^c}}} 
    \ket{JJ^{*},a,b}_{AA^c} \bra{JJ^{*},a,b}_{AA^c}
\end{align}
Note, the bipartite entanglement properties of $\ket{JJ^{*},a,b}_{AA^c}$ are independent of the multiplicity labels $a$ and $b$, and are just the same as the singlet $\ket{JJ^{*}}_{AA^c}$. 

The projector on the singlet $ [\ket{JJ^{*}}\bra{JJ^{*}}]_{AA^c}$ can be expressed in terms of an integral over the tensor product of unitary representations of the group on $A$ and $A^c$ respectively, $\sigma = \int_{G}dg U_{J}(g) \otimes {U_{J^*}(g)}$, with the Haar measure $\int_{G} dg = 1$ (with an equivalent description for finite groups, replacing the integral by the average over group elements) \footnote{This is true, since the property of singlets $\ket{JJ^*}_{AA^c}$ dictates that $U_J(g)\otimes U_{J^*}(g)\ket{JJ^*}_{AA^c} = \ket{JJ^*}_{AA^c}$ for all $g\in G$.}. Let the dimension of irrep $J$ be denoted as $d_{J}$. The Schur orthogonality relation for the group elements for any compact group leads to the following expression for the elements of the projector, 
\begin{align}\label{eq:Schur_orthog_singlets}
        \left[\ket{JJ^{*}}\bra{JJ^{*}}\right]_{ij,kl} &= \int_{G}dg [U_{J}(g)]_{ij} \otimes [U_{J^*}(g)]_{kl} = \frac{\delta_{ik}\delta_{jl}}{d_{J}} \nonumber\\&\quad \quad \text{for }i,j\text{ in }A, \text{and }k,l\text{ in }A^c.  
    \end{align}
Eqs.~\eqref{eq:mms_decomp} and ~\eqref{eq:Schur_orthog_singlets}
 allow us to compute various entanglement quantities of the MMIS exactly.

\subsection{Entanglement of formation and distillation}\label{sec:entanglement_equality}
\begin{figure}
    \centering
    \includegraphics[width = \columnwidth]{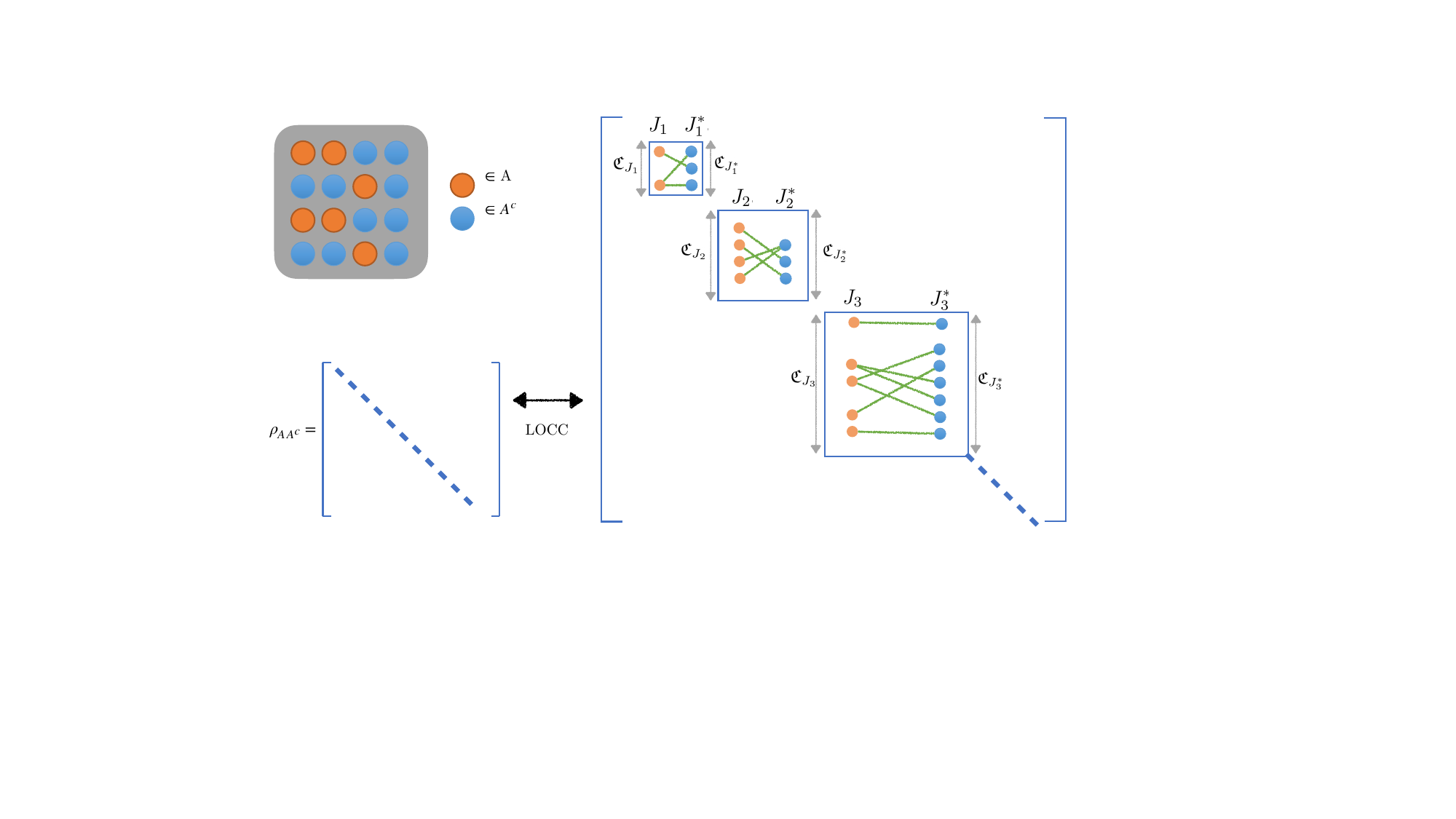}
    \caption{Block diagonal bipartite structure of $\rho_{\triv}^N$ into two (not necessarily contiguous) parts $A$ and $A^c$. The state can be decomposed into singlets between $J_i$ irrep in $A$ and $J_i^*$ irrep in $A^c$, which appear in $\mathfrak{C}_{J_i}$ and $\mathfrak{C}_{J_i^*}$ times. The pairings can be distinguished by LOCC, without collapsing the entanglement in the pure states. }
    \label{fig:Block_diag}
\end{figure}
We show here that the bipartite entanglement measures for the maximally mixed invariant state match, and can be easily computed, starting from the convenient decomposition in Eq.~\eqref{eq:mms_decomp}. An upper-bound for the entanglement of formation is directly given by the specific decomposition in Eq.~\eqref{eq:mms_decomp}, 
\begin{align}
    E^{f,*}_{A:A^c}(\rho_{\triv}^{N}) = \frac{1}{\mathfrak{C}_{\triv}^{N}}\sum_{J\in \mathcal{J}_{AA^c}}\mathfrak{C}_{J}^{N_{A}}\mathfrak{C}_{J^{*}}^{N_{A^c}}S_{A:A^c}(\ket{JJ^{*}}_{AA^c}).
\end{align}

For the entanglement of distillation, Alice and Bob can individually measure the irrep $J$ and $J^{*}$ and the local multiplicity labels $a,b$, with access to just their qubits. Thus they can obtain the pure singlet state $\ket{JJ^{*},a,b}$ with probability $1/\mathfrak{C}^{N}_{\triv}$ for any $a,b$ and for the allowed irreps $J \in \mathcal{J}_{AA^c}$. Each singlet state $\ket{JJ^{*},a,b}$ can be distilled to give $S_{A:A^c}(\ket{JJ^{*}})$ Bell pairs~\cite{Bennett_1996}, and thus by this distillation strategy, the lower bound to the entanglement of distillation is also given by,
\begin{align}\label{eq:mms_entanglement}
    E^{D,*}_{A:A^c}(\rho_{\triv}^{N}) = \frac{1}{\mathfrak{C}_{\triv}^{N}}\sum_{J\in \mathcal{J}_{AA^c}}\mathfrak{C}_{J}^{N_{A}}\mathfrak{C}_{J^{*}}^{N_{A^c}}S_{A:A^c}(\ket{JJ^{*}}_{AA^c}).
\end{align}

Since the lower bound to $E^{D}_{A:A^c}$ and the upper bound to $E^{F}_{A:A^c}$ match, by the hierarchy of entanglement measures in Eq.~\eqref{eq:ent_hierarchy}, all \textit{good} entanglement measures are equal, and equal the expression in Eq.~\eqref{eq:mms_entanglement}. The bipartite structure of the density matrix which leads to this result is described schematically in Fig.~\ref{fig:Block_diag}.

The crucial property of the state that allows one to derive this result is that $\rho_{\triv}^{N}$ is a mixture of pure states $\ket{JJ^{*},a,b}$ which are distinguishable by LOCC, without destroying the entanglement between $A$ and $A^c$~\cite{livine_entanglement_2005}. This is a generalization of the property of `local orthogonality' property of mixed states introduced in~\cite{horodecki_entanglement_1998}, who demonstrated that mixtures of locally orthogonal states have equal entanglement of formation and distillation.

We can further simplify the expression of the entanglement by explicitly computing the entanglement of the singlet $\ket{JJ^{*}}_{AA^c}$, which is obtained form the reduced density matrix on subregion $A$ resulted from tracing out its complement $A^c$. Here we use the Schur orthogonality relations in Eq.~\eqref{eq:Schur_orthog_singlets} for the partial trace on the projector on singlets $\Tr_{A^c}\ket{JJ^{*}}\bra{JJ^{*}} = (1/d_{J})\sum_{\alpha = 1}^{d_{J}}\ket{\alpha}\bra{\alpha} = \mathbb{I}_{J}/d_{J}$, with $\mathbb{I}_{J}$ being the $d_J\times d_J$ dimensional identity matrix. This implies that $S_{A:A^c}(\ket{JJ^{*}}) = \log d_J$. Thus we have the simplified expression for the entanglement of formation and the entanglement of distillation, 
\begin{align}\label{eq:mmis_entanglement}
    E_{A:A^c}\left(\rho_{\triv}^{N}\right) = \frac{1}{\mathfrak{C}_{\triv}^{N}}\sum_{J\in \mathcal{J}_{AA^c}}\mathfrak{C}_{J}^{N_{A}}\mathfrak{C}_{J^{*}}^{N_{A^c}} \log d_{J}.
\end{align}

Note, the equality of the entanglement measures for the maximally mixed invariant state was earlier shown for $SU(2)$ group and tensor product spaces of their fundamental representations by~\cite{livine_entanglement_2005}, using Schur-Weyl duality. Here the result is shown to be true for any compact group: independent of the specific group, tensor product space, or Schur-Weyl duality.


\subsubsection{Mixed states in other symmetry sectors}
\label{sec:other_sym_sectors}




We briefly comment on whether mixed states in other symmetry sectors have the property of the equality of entanglement measures. In the invariant sector of the tensor product space of irreps, the factorization of multiplicities in Eq.~\eqref{eq:multiplicity_factorization}, $\mathfrak{C}^{N}_{\triv} = \sum_{J \in \mathcal{J}_{AA^c}}\mathfrak{C}_{J}^{N_A}\mathfrak{C}_{J^{*}}^{N_{A^c}}$ is an equivalent condition to the states being locally distinguishable by LOCC, without destroying the entanglement. However, in general, the Clebsch-Gordan coefficients $\mathfrak{D}^{J}_{J_A J_{A^c}}$ can be greater than 1 for certain symmetries and their representations. Thus, in general, the equality of entanglement measures is a specific property of mixed states in the invariant sector. 

In the particular case of $SU(2)$ which is multiplicity-free~\cite{hall_lie_2015}, it turns out that the Clebsch-Gordan multiplicities are trivial $\mathfrak{D}^{J}_{J_A J_{A^c}} = 0, 1$ for all irreps and symmetry sectors. The irreps are the spin$-j$ representations of $SU(2)$, and the states are specified by the irrep label, the total angular momentum $j$ and the angular momentum along $z$ direction $m$. We can consider equiprobable mixture of pure states within such $(j,m)$ sectors. Recall, they arise naturally under the application of strong symmetric channels on an initial pure state in a specific $(j,m)$ sector, as described in Sec.~\ref{sec:channel_thm}. The pure states appearing in the mixed state are thus of the form (given a bipartition $A:A^c$),  $\ket{j_a j_b; j m, a, b}$, where $j_a$ and $j_b$ refer to the total angular momenta in $A$ and $A^c$ respectively, $j, m$ refer to the angular momenta of the whole system, and $a,b$ refer to the multiplicity labels in $A$ and $A^c$ respectively. These states are locally distinguishable by LOCC, by just measuring the labels $j_a$ and $a$ in $A$, and $j_b$ and $b$ in $B$, followed by classical communication. In particular, the entanglement of formation and distillation of the maximally mixed state within a particular $(j,m)$ sector for the $SU(2)$ group is given by, 
\begin{align}\label{eq:su2_entanglement}
&E_{A:A^c}\left(\rho_{(j,m)}^{N}\right) = \nonumber\\
&\frac{1}{\mathfrak{C}_{j}^{N}}\sum_{j_a, j_b}^{|j_a-j_b| \leq j \leq j_a + j_b}\mathfrak{C}_{j_a}^{N_{A}}\mathfrak{C}_{j_b}^{N_{A^c}} S_{A:A^c}\left(\ket{j_a j_b; j m}_{AA^c}\right). \nonumber\\
&\quad \quad \quad \quad \text{(specific to SU(2))}
\end{align}



\subsection{Symmetry-based classification}\label{sec:entanglement_classification}
Eq.~\eqref{eq:mmis_entanglement} implies that, generically, non-Abelian symmetries lead to entangled MMIS, since there exist irreps with $d_J>1$. In contrast, Abelian symmetries necessitate zero entanglement of formation, and hence, separability of the MMIS. 

If the non-Abelian symmetry is a finite discrete group, Eq.~\eqref{eq:mmis_entanglement} implies the entanglement of formation for any bipartition is $O(1)$, even in the thermodynamic limit, $N\to \infty$. We provide an explicit example of this in Sec.~\ref{sec:s3_entanglement} for the $S_3$ MMIS for a spin chain. 

The only case in which the MMIS of a non-Abelian symmetry is separable with respect to a bipartition $A:B$ is when $d_J = 1$ for all allowed irreps $J \in \mathcal{J}_{AB}$. This happens, for example, if the on-site Hilbert space $\mathcal{V}$ is the 2d irrep of $G = D_4$, and if both $A$ and $B$ have even sites, as one can check from the character table of $D_4$. This is, however, always a fine-tuned scenario, as adding one more site to both $A$ and $B$ gives at least one higher-dimensional irrep in $\mathcal{J}_{AB}$ no matter the symmetry group and $\mathcal{V}$. Indeed, we prove in Appendix \ref{appsec:exp_values_proof} that there are infinite sequences of $N$ and $|A|$ for which the MMIS is entangled.

In Sec.~\ref{sec:lie_mmis}, we consider the case of non-Abelian continuous groups, which are characterised by irreps of dimension $d_J$ that are extensive in the thermodynamic limit. In Sec.~\ref{sec:lie_mmis_su2}, we first show that the half system entanglement for the $SU(2)$ MMIS on $N$ qubits is $1/2 \log N +O(1)$ in the thermodynamic limit $N \gg 1$. In Sec.~\ref{sec:lie_mmis_gen} we generalize the results to any compact semisimple Lie group, to show that the entanglement of formation and distillation scales as $\Theta(\log N)$ for large enough bipartitions ($A, A^c \gtrsim O(\sqrt{N})$). 

This implies that MMIS for such groups are genuinely long-range entangled, and we argue in Sec.~\ref{sec:entanglement_measures}, that at least a $\Omega(\log N)$ depth adaptive circuit is necessary to prepare such a state. 



\subsection{Other measures: Negativity and Operator Entanglement}

Equation \eqref{eq:mms_decomp} also allow us to exactly calculate the entanglement negativity of $\rho_{\triv}^N$. First, note that the density matrices of the singlet states $\ket{JJ^{*},a,b}_{AA^c}$ are supported in orthogonal subspaces for different $J$, $a$ or $b$. This means that the negativity $\mathcal{N}_{A:A^c}(\rho) = (\norm{\rho^{T_A}}_1-1)/2$ of $\rho_{\triv}^N$ is just the convex sum of the negativities for each singlet state:
\begin{equation}
    \mathcal{N}_{A:A^c}(\rho_{\triv}^{N}) = \frac{1}{\mathfrak{C}_{\triv}^{N}}\sum_{J\in \mathcal{J}_{AA^c}}\mathfrak{C}_{J}^{N_{A}}\mathfrak{C}_{J^{*}}^{N_{A^c}} \mathcal{N}_{A:A^c}(\ket{JJ^{*}}_{AA^c}).
\end{equation}
Moreover, from Schur's orthogonality relations \eqref{eq:Schur_orthog_singlets}, we know that each singlet is the maximally entangled state in a Hilbert space of dimension $d_J$, which has negativity $(d_J-1)/2$. Hence, the logarithmic negativity of $\rho_{\triv}^{N}$ is given by 
\begin{equation}
    E^\mathcal{N}_{A:A^c}(\rho_{\triv}^{N}) = \log \norm{\rho^{T_A}}_1 = \log \sum_{J\in \mathcal{J}_{AA^c}} \frac{\mathfrak{C}_{J}^{N_{A}}\mathfrak{C}_{J^{*}}^{N_{A^c}}}{\mathfrak{C}_{\triv}^{N}} d_J.
\end{equation}
In particular, from the concavity of $\log$, we recover the general result that $E^\mathcal{N}_{A:A^c} \geq E^\mathcal{D}_{A:A^c}$. Furthermore, as we will show in the following sections, for continuous groups the probability distribution in the convex sum above is highly peaked around a particular $J_{\text{max}}$ in the thermodynamic limit, implying the logarithmic negativity also converges to the same value as the \textit{good} entanglement quantities given by \eqref{eq:mmis_entanglement}. We also note that a similar calculation can be applied to calculate all Rényi negativities.


Equation \eqref{eq:mms_decomp} also allows us to exactly calculate several other entanglement-like quantities of $\rho^{N}_{\triv}$. Taking the partial trace over $A^c$ in Equation \eqref{eq:mms_decomp}, 
\begin{align}\label{eq:reduced_MMIS}
    \rho^{A}_{\triv} \equiv \Tr_{A^c} \rho^{N}_{\triv} = \sum_{J\in \mathcal{J}_{AA^c}}\sum_{a = 1}^{\mathfrak{C}^{N_A}_{J}}\frac{\mathfrak{C}^{N_{A^c}}_{J^*}}{\mathfrak{C}_{\triv}^{N}d_{J}}\mathbb{I}_{J}\otimes \ket{a}\bra{a}.
\end{align}

Computing the von Neumann entropy of $\rho^{A}$ is now simple, since the density matrix is block diagonal with constant entries. We get the following expression for the entanglement entropy,
\begin{align}\label{eq:ent_entropy}
    S(\rho^{A}_{\triv}) = \sum_{J \in \mathcal{J}_{AA^c}}\frac{\mathfrak{C}^{N_A}_{J}\mathfrak{C}^{N_{A^c}}_{J^*}}{\mathfrak{C}_{\triv}^{N}}\log\left(\frac{\mathfrak{C}_{\triv}^{N}d_{J}}{\mathfrak{C}_{J^*}^{N_{A^c}}}\right).
\end{align}

We can also compute the operator entanglement of the maximally mixed state exactly as well. Operator entanglement is a measure of both quantum and classical correlations in the density matrix, and is defined as the entanglement entropy of the vectorized density matrix. For the maximally mixed invariant state, we have, $|\rho^{N}_{\triv}) = \frac{1}{\sqrt{\mathfrak{C}^{N}_{\triv}}}\sum_{J \in \mathcal{J}_{AA^c}}\sum_{a,b}\ket{JJ^{*}ab}\ket{JJ^{*}ab}$, where the prefactor is added to normalize the vectorized density matrix. The operator entanglement is defined as the von Neumann entropy of reduced density matrix constructed from the vectorized density matrix, $\Tr_{A^c}|\rho^{N}_{\triv})(\rho^{N}_{\triv}|$. Using the Schur orthogonality relations for the projector on singlets and the bipartite decomposition in Eq.~\eqref{eq:mms_decomp}, we get,
\begin{align}
    &\Tr_{A^c}|\rho^{N}_{\triv})(\rho^{N}_{\triv}| = \nonumber\\&\frac{1}{\mathfrak{C}_{\triv}^{N}}\sum_{J \in \mathcal{J}_{AA^c}}\frac{\mathfrak{C}_{J^*}^{N_B}}{d_J^2}\mathbb{I}_{J}\otimes \mathbb{I}_{J}\otimes \sum_{a,a^{\prime} = 1}^{\mathfrak{C}_{J}^{N_A}}\ket{aa}\bra{a^{\prime}a^{\prime}}.
\end{align}
This formula reveals that this matrix is block diagonal into identity blocks with respect to the irreps and the matrix with all equal elements with respect to the multiplicities in $A$. The spectrum of such a matrix is simple, which leads to the following expression for the bipartite operator entanglement, 
\begin{align} \label{eq:op_ent}
    &O_{A:A^c}\left(\rho^{N}_{\triv}\right) = \sum_{J \in \mathcal{J}_{AA^c}}\frac{\mathfrak{C}^{N_A}_{J}\mathfrak{C}^{N_{A^c}}_{J^*}}{\mathfrak{C}_{\triv}^{N}}\log\left(\frac{d_J^{2}\mathfrak{C}_{\triv}^{N}}{\mathfrak{C}_{J}^{N_A}\mathfrak{C}_{J^*}^{N_{A^c}}}\right) \nonumber\\&= 2E_{A:A^c}\left(\rho^{N}_{\triv}\right) + \sum_{J \in \mathcal{J}_{AA^c}}\frac{\mathfrak{C}^{N_A}_{J}\mathfrak{C}^{N_{A^c}}_{J^*}}{\mathfrak{C}_{\triv}^{N}}\log\left(\frac{\mathfrak{C}_{\triv}^{N}}{\mathfrak{C}_{J}^{N_A}\mathfrak{C}_{J^*}^{N_{A^c}}}\right).
\end{align}

The above equation explicitly shows that the operator entanglement involves both quantum and classical contributions, thus is not a good measure of quantum entanglement. However, the operator entanglement can still be useful. For example, it is related to bond dimension of the matrix product density operator (MPDO) representation of the density matrix by $D = e^{O_{A:A^c}}$. Eq.~\eqref{eq:op_ent} implies that the MPDO bond dimension of an MMIS scales like $O(1)$ and $O(poly(N))$ for finite and continuous non-Abelian groups respectively, meaning that such mixed states can be efficiently simulated with tensor network algorithms, regardless of spatial dimensions. We discuss MPDO representations more carefully in Appendix \ref{appsec:MPO}. We will also show that MMISs can be purified to a matrix product state (MPS) with the purification rank~\cite{cirac_MPDO_BondDim} $D' \leq D$. The difference between MPDO representations of the thermal state in canonical (strong symmetry) and grand canonical (weak symmetry) ensembles was also studied in~\cite{Barthel_2016}.      

\section{Discrete non-Abelian symmetry: $S_{3}$} \label{sec:s3_entanglement}

As an example for discrete non-Abelian groups, we study the MMIS of the symmetric group $S_3$. Generally, the symmetric group $S_k$ is the group of all permutations of $k$ objects. For the case of $S_3$, it is also equivalent to the dihedral group $D_3$ which is the symmetry group of an equilateral triangle on a surface (say the $xy$-plane), whose elements are 3-fold rotations about the $z$-axis, and three reflections with respect to the medians of the triangle.  

$S_3$ has three irreps: two 1 dimensional irreps, namely the trivial (\triv) and the sign (\textsf{sgn}) representations, and one 2 dimensional standard representation (\twod). In the \textsf{sgn} representation, we assign $+1$ to rotation elements and $-1$ to reflection ones. There is a natural representation for the \textsf{2d} representation in terms of Pauli operators of a qubit: we can take the generators to be $R = e^{\frac{2\pi i}{3} Z}$ and $P_x = X$, where $X$ and $Z$ are Pauli matrices. The states in the \textsf{2d} irrep can be labeled by eignenvalues of rotation, $R \ket{\theta = \pm 2\pi/3} = e^{i \theta} \ket {\theta}, \: P_x \ket{\theta} = \ket{-\theta} $,
which in the above choice of the basis are $\ket{\theta = 2\pi/3} = \ket{\uparrow}$ and $\ket{\theta = -2\pi/3} = \ket{\downarrow}$. 

Consider a spin chain of $N$ qubits in the $\twod$ irrep of $S_3$, and the corresponding MMIS on this spin chain. The Clebsch-Gordan series of the tensor product Hilbert space decomposes as,
\begin{align}
    \mathcal{H}_{\twod}^{\otimes N} = \mathfrak{C}^{N}_{\triv}\mathcal{H}_{\triv}\oplus \mathfrak{C}^{N}_{\textsf{sgn}}\mathcal{H}_{\textsf{sgn}} \oplus \mathfrak{C}^{N}_{\twod}\mathcal{H}_{\twod}.
\end{align}

The multiplicities can be easily computed using the Schur orthogonality relations for the characters of the group (similarly as the strategy from Eq.~\eqref{eq:cg_invariant}), and we obtain,

\begin{align}\label{eq:s3_multiplicities}
\mathfrak{C}^{N}_{\triv} &= \mathfrak{C}^{N}_{\textsf{sgn}} = \frac{1}{6} 2^N + \frac{1}{3} (-1)^N\nonumber\\\mathfrak{C}^{N}_{\twod} &= \frac{1}{3} 2^N - \frac{1}{3} (-1)^N.
\end{align}

Applying these results to the expression  Eq.~\eqref{eq:mmis_entanglement}, we find that the half chain bipartite entanglement of formation and distillation for of the $S_3$ MMIS in the thermodynamic limit,
\begin{align}
E_{N/2:N/2}\left(\rho_{\triv}^{N}\right) \overset{N\to\infty}{\longrightarrow} \frac{2\log{2}}{3},  
\end{align} 
while the corresponding logarithmic negativity is given by,
\begin{align}
E^{\mathcal{N}}_{N/2:N/2}\left(\rho_{\triv}^{N}\right) \overset{N\to\infty}{\longrightarrow} \log{\frac{5}{3}} \quad \quad 
\end{align}

The fact that that the entanglement remains $O(1)$ in the thermodynamic limit suggests that the state is a mixture of short-range entangled states. We probe this in two ways. Firstly, we show in Appendix Sec.~\ref{appsec:s3_correlations} that the two point function for an operator charged under $S_3$, such as $Z$ operator is exponentially suppressed, $\Tr\left[\rho^{N}_{\triv}Z_i Z_j\right] \sim \exp(- \alpha N)$. This observation shows it is short-range correlated with respect to this operator. Furthermore, in Appendix Sec.~\ref{appsec:s3_channel} we explicitly construct a $S_3$ symmetric quantum channel with unitary Kraus operators, and numerically demonstrate that such a channel generates this steady state in $O(1)$ time. This proves that such a state is indeed short-range entangled. 


\begin{figure}
    \centering
    \includegraphics[width = 0.8\columnwidth]{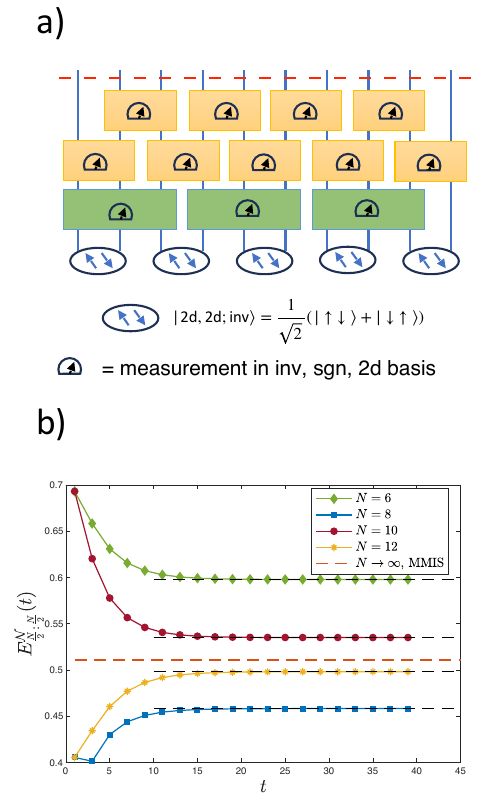}
    \caption{(a) $S_3$ symmetric channel with an open boundary condition. Each gate is a measurement in the basis of the total $S_3$ charges ($\triv$, $\sgn$, and $\twod$) of the qubits on which the gate is applied. The system is initialized with the nearest neighbor pairs in the $\triv$ sector. One can check that the invariant singlet state of two qubits in the $\{\ket{\uparrow},\ket{\downarrow}\}$ basis introduced in Sec.~\ref{sec:s3_entanglement} is given by $\frac{1}{\sqrt{2}} (\ket{\uparrow \downarrow} + \ket{\downarrow \uparrow}$). Each time step is identified with every 3 measurement layers. The steady state of this channel is the $S_3$ MMIS as discussed in Appendix \ref{appsec:s3_channel} (b) The half-chain logarithmic negativity as a function of time for different system sizes. The black dashed lines are the analytical steady states values of the logarithmic negativity. The plot indicates that the MMIS negativity remains finite as the system size increases.}
    \label{fig:S3_neg_vs_time}
\end{figure}




\section{Continuous non-Abelian symmetry}\label{sec:lie_mmis}
\subsection{$SU(2)$}\label{sec:lie_mmis_su2}

In the $SU(2)$ case, the multiplicities of the different spin$-j$ representations in a chain of qubits are well known, for example, via the Schur-Weyl duality. In particular, in a chain of $2N$ qubits, the multiplicities are given by~\cite{livine_entanglement_2005, livine_quantum_2006},
\begin{align}\label{eq:su2_multiplicity}
\mathfrak{C}_{j}^{2N} = {2N \choose N + j}\frac{2j+1}{N+j+1}.
\end{align}

By an exact computation in the thermodynamic limit $N\gg 1$ using Eq.~\eqref{eq:mmis_entanglement} and saddle point approximation, we find that the half-system entanglement (entanglement of formation, distillation, and logarithmic negativity) is,
\begin{align}
E_{N:N}\left(\rho_{\triv}^{2N}\right) \overset{N\to\infty}{\longrightarrow} \frac{1}{2}\log N + O(1)
\end{align}

A generalization of the result for $SU(d)$ groups using Schur-Weyl duality is provided in Appendix~\ref{appsec:su(d)}.

\subsection{Compact Semisimple Lie groups}\label{sec:lie_mmis_gen}
The bipartite entanglement of formation of the maximally mixed invariant state can be evaluated for generic symmetries described by compact semisimple Lie groups as well. We can rewrite the expression for the bipartite entanglement of formation as a weighted sum, $E_{A:B}\left(\rho_{\triv}^{N}\right) = \sum_{J\in \mathcal{J}_{AB}}p(J) \log d_{J}$, with the probability distribution, $p(J) = \frac{\mathfrak{C}_{J}^{N_{A}}\mathfrak{C}_{J^{*}}^{N_{B}}}{\mathfrak{C}_{\triv}^{N}}$. This is a probability distribution, since, by definition, $\sum_{J\in \mathcal{J}_{AB}}p(J) = 1$. In the thermodynamic limit, for a bipartition such that $N_{A}, N_{B},N\gg 1$, we can approximate the weighted sum by the most likely value from saddle point analysis, $E_{A:B}\left(\rho_{\triv}^{N}\right)\approx \log d_{\overline{J}}$ where $\overline{J}$ maximizes the probability distribution $ \overline{J} = \text{argmax}_{J \in \mathcal{J}_{AB}}p(J)$.

First we discuss the asymptotic scaling of irrep dimensions. Suppose the dimension and the rank of the Lie algebra corresponding to the Lie group be $\dim \mathfrak{g}$ and $\dim \mathfrak{h}$ respectively. Assume the dimension of the local Hilbert space $\mathcal{H}_{d}$ be $d$, which carries an irreducible representation of the Lie group. By the weight theory of irreps of complex semisimple Lie algebra~\cite{hall_lie_2015}, every irrep is uniquely labeled by its highest weight $w$, which is a vector in the vector space of its Cartan subalgebra. The weight $w$ is represented as an $M$ dimensional vector; and the highest weight corresponding to an irrep is a vector of non-negative integers. The tensor product space of $N$ $d$-dimensional irreps, $\mathcal{H}_{d}^{\otimes N}$, can be decomposed into a tensor sum of irreps. We are interested in the asymptotic scaling of dimensions of irreps in the tensor sum, in the thermodynamic limit $N\gg 1$. The largest weight of an irrep that can be accommodated in this tensor product space has elements scaling as $O(N)$. 

We now show that if the elements of the weight vector scales asymptotically as $O(n)$, the dimension of the irrep scales as $O(n^{(\dim \mathfrak{g}-\dim \mathfrak{h})/2})$. We can justify this using the fact that the number of positive roots of the root system associated with the Lie group is $(\dim \mathfrak{g}-\dim \mathfrak{h})/2$. By the Weyl dimension formula, the dimension of the irrep is given by a polynomial of degree equal to the number of positive roots, i.e. $(\dim \mathfrak{g}-\dim \mathfrak{h})/2$. When the elements of the weight vector scale as $O(n)$, the dimension thus scales polynomially as $O((\dim \mathfrak{g}-\dim \mathfrak{h})/2)$.
\begin{figure}
    \centering
    \includegraphics[width = \columnwidth]{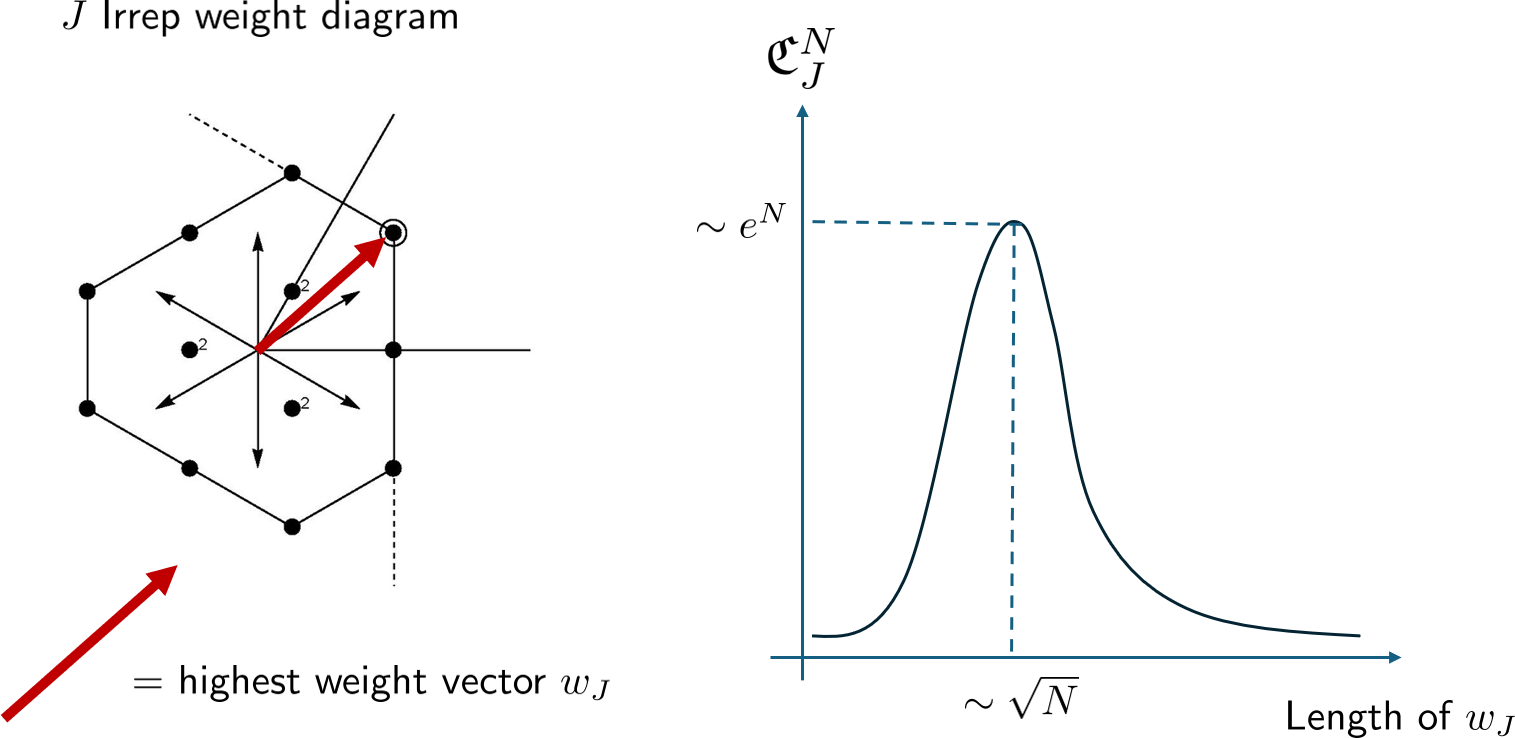}
    \caption{Any irrep $J$ of compact semisimple Lie algebras can be represented using a finite dimensional lattice of weights in a $\dim \mathfrak{h}$ dimensional space, where $\dim \mathfrak{h}$ is the rank of the Lie algebra. Each irrep is uniquely characterized by a highest weight vector, $w_J$. In the tensor product space $V_d^{\otimes N}$, the multiplicity of $J$ irrep, $\mathfrak{C}^{N}_{J}$ peaks when the length of the $w_J$ vector scales as $\sim\sqrt{N}$~\cite{biane_estimation_1993, tate_lattice_2003, feigin_large_2019}, in the thermodynamic limit $N\to \infty$.}
    \label{fig:gen_lie}
\end{figure}
Next, we show that the probability of the multiplicities, $p(J)$ is maximized for weights $\sim \sqrt{N}$, when $N_A, N_B \gg O(\sqrt{N})$. Here we are interested in the multiplicity $\mathfrak{C}_{J}^{N}$, for irrep $J$ whose weight vector elements scale asymptotically as $~n$ for some large integer $n$, and where $N\gg 1$. As shown in several mathematics papers~\cite{biane_estimation_1993, tate_lattice_2003, feigin_large_2019}, near the maxima of this multiplicity, the asymptotic scaling is a Gaussian of the form $\mathfrak{C}_{J\sim n}^{N} \sim O\left(\text{poly}(n)\exp(-n^{2}/N)\right)$. This is maximized for $n\sim \sqrt{N}$ asymptotically. This result can be understood intuitively by thinking of the representation as a random walk on the weight lattice. A schematic of the highest weight vector and the scaling of the multiplicity in the tensor product space is shown in Fig.~\ref{fig:gen_lie}. Combining all the results, we get that the asymptotic result for the bipartite entanglement of formation for $N = N_{A} + N_B$ bipartition, where $N_{A},N_B,N \gg 1$ and without loss of generality, $N_A \leq N_B$, 
\begin{align}
    E_{N_A:N_B}(\rho^{N}_{\triv}) &\approx \log \sqrt{N_A}^{(\dim \mathfrak{g}-\dim \mathfrak{h})/2} + O(1)\nonumber \\
    &= \frac{\dim \mathfrak{g}-\dim \mathfrak{h}}{4}\log{N_A} + O(1). \label{eq:log_EF}
\end{align}

Note, for $SU(d)$ groups, we have $\dim \mathfrak{g} = d^2 -1$ and $\dim \mathfrak{h} = d-1$, which leads to the asymptotic result $\frac{d(d-1)}{4}\log N$. In Appendix~\ref{appsec:su(d)} we provide an alternate derivation of this result using explicit computation of the multiplicities and dimensions of the irreps of $SU(d)$ in a chain of $d$ dimensional qudits, using tools from Schur-Weyl duality. In the next section, we discuss the $SU(2)$ case in more detail.

\subsection{Entanglement at finite temperature}\label{sec:thermal_mmis}

The maximally mixed invariant states can be viewed as infinite-temperature thermal states subject to constraints on total symmetry charges (for $U(1)$ symmetry this would be usual ``canonical ensemble''). Given a Hamiltonian $H$, we can consider such invariant thermal states at finite temperature
\begin{equation}
    \rho_T=\frac{1}{Z}e^{-\frac{H}{T}}P_{\triv},
\end{equation}
where $P_{\triv}$ is the projector to the symmetry invariant subspace. For a generic Hamiltonian, the thermal state is complicated to analyze, and it is unclear whether a finite-temperature state will be long-range entangled with $O(\log N)$ entanglement. Below, we will consider a solvable example, which has only $O(1)$ entanglement unless $T\gtrapprox O(N)$ (the ``infinite temperature'' regime). 

The example we consider is an $SU(2)$ spin system with a spin-$S$ moment per site, with the Hamiltonian
\begin{equation} \label{eq:Hamiltonian}
    H=(\sum_{i\in A}S_i)^2+(\sum_{j\in A^c}S_j)^2,
\end{equation}
where $A,A^c$ is a bi-partition of the system. The model is designed so that the thermal states take a simple form similar to Eq.~\eqref{eq:mms_decomp}:
\begin{equation}
    \rho_T\propto\sum_{J,a,b} e^{-\frac{2J(J+1)}{T}}
    \ket{JJ^{*},a,b}_{AA^c} \bra{JJ^{*},a,b}_{AA^c}.
\end{equation}
In particular, the state can be written as the convex sum of locally distinguishable pure states. This allows the arguments in Sec.~\ref{sec:entanglement_equality} to apply, from which we conclude that the entanglement of formation and distillation are equal and given by
\begin{align}
\label{eq:entanglement_finiteT}
    E_{A:A^c}\left(\rho_{\triv}^{N}\right) = \frac{1}{\mathcal{Z}}\sum_{J\in \mathcal{J}_{AA^c}}e^{-\frac{2J(J+1)}{T}}\mathfrak{C}_{J}^{N_{A}}\mathfrak{C}_{J^*}^{N_{A^c}} \log d_{J},
\end{align}
where $\mathcal{Z} = \sum_{J\in \mathcal{J}_{AA^c}}e^{-\frac{2J(J+1)}{T}}\mathfrak{C}_{J}^{N_{A}}\mathfrak{C}_{J^*}^{N_{A^c}}$. We can now consider two different temperature regimes:
\begin{enumerate}
    \item $T\gtrapprox \Omega(N)$ (the ``infinite temperature'' regime): similar to the maximally mixed invariant states, the sum in Eq.~\eqref{eq:entanglement_finiteT} is dominated by $J\sim \sqrt{N}$, so $E_{A:A^c}\sim \log N$.
    \item $T\ll N$ (the ``finite temperature'' regime): the sum in Eq.~\eqref{eq:entanglement_finiteT} is cutoff at $J\sim \sqrt{T}$, so we should have $E_{A:A^c}\sim \log T$, which does not scale with $N$.
\end{enumerate}
Our result suggests that infinite-temperature invariant states form a universality class distinct from any finite-temperature states. The intuition is that large entanglement from large spin states will be exponentially suppressed at any finite-temperature. Even though the particular Hamiltonian Eq.~\eqref{eq:Hamiltonian} is by no means generic (it is highly fine-tuned, non-local, and not translationally invariant), our conclusion may hold more generally. 



Our discussion makes the MMIS similar to the usual quantum criticality at zero temperature: if a $(1+1)d$ conformal field theory emerges at zero temperature, the entanglement scales with $\log N$, but at any  finite temperature the entanglement negativity scales as $O(1)$\cite{CardyNegativity}.

Another interesting question is whether there is some special structure -- such as translation invariance and Lieb-Schultz-Mattis type of constraints -- that will guarantee (at least) $\log N$ entanglement at finite temperature. We leave this question for future study.


\section{Preparation and stability of MMIS}

The long-range correlations and the entanglement structure of MMIS of continuous non-Abelian symmetries lead us to consider two (related) questions:

1) How easy is it to prepare such a long-range entangled mixed states from an unentangled pure product state?

2) How stable is such long-range entanglement to perturbations on MMIS?

In the following sections, we answer both of these questions. In Sec. \ref{sec:preparation}, we bound the preparation time by $O(L)$ under local channels due to long-range correlations decaying slower than exponential, and under local unitaries and measurements with nonlocal feedback by $O(\log L)$ from the entanglement results. In Sec. \ref{sec:stability_swssb}, we prove that the MMIS is part of a stable symmetry protected mixed phase characterized by a novel type of SSB, known as strong-to-weak SSB.

\subsection{Bounds on preparation time}\label{sec:preparation}

We consider two strategies to prepare such a state.

\textbf{Local channel preparation: }In these protocols, we consider evolution with local channels (which can be a channel with brickwork Kraus operators, or a Lindbladian evolution with local Jump operators). In such a scenario, the
Lieb-Robinson bound~\cite{lieb_finite_1972} holds, which limits how quickly correlations can build under local evolution~\cite{Bravyi_2006}. The Lieb-Robinson bound was originally formulated for local unitary evolution; however, it can be generalized to local channels~\cite{Poulin_2010}. The same bound holds in general, since the local channel can be purified into a local unitary operation followed by tracing out of the local environment. Under such evolutions, the connected correlation function is bounded as~\cite{Bravyi_2006},
\begin{align}
    \langle O_{i}O_j \rangle_c \leq c \exp[\frac{vt-d_{ij}}{\xi}],
\end{align}
for $O(1)$ constant $c$, Lieb-Robinson velocity $v$, and correlation length $\xi$. In Sec.~\ref{sec:mms_correlations} we showed that the connected correlations between the Lie algebra generators scale as $1/N$, for local qudits separated by distance $N$. This immediately implies that the time-scale of preparing such a state using local strategies scales at least linearly with $N$, i.e. $\Omega(N)$.


\textbf{Local adaptive preparation:} Since we established that the MMISs are also generically highly entangled, we can make a stronger statement about any local adaptive operations to produce such a state from pure product states. By local adaptive operations, we refer to finite-depth circuits with local unitaries and measurements and possibly nonlocal classical communication. Due to nonlocal communication, quantum and classical correlations at long distances can be established, as examplified by the quantum teleportation protocol. This in turn invalidates our previous argument based on connected correlation functions. Nevertheless, we will prove that the entanglement of formation of quantum many-body mixed states serves as a lower bound for the depth (or time) required to prepare a certain state starting from an product state via local operations and classical communication.\footnote{We note that \cite{lu_mixed-state_2023} used a similar argument to bound the R\'enyi entanglement entropy of formation by the adaptive preparation depth.}

To see this, suppose first there is no measurement and no classical communication. Then, we return to the scenario considered above with local channels, where there is a local quantum circuit $U$ with depth $D$ acting on the system plus ancilla $\mathcal{H} \otimes \mathcal{H}'$ such that $\rho = \Tr_{\mathcal{H}'}[U \ketbra{\psi_0}{\psi_0} U^\dagger]$, where $\ket{\psi_0}$ is any product state in $\mathcal{H} \otimes \mathcal{H}'$. Given a bipartition $A:B$ of $\mathcal{H}$, then from the locality of $U$ we know that entanglement entropy of $U \ket{\psi_0}$ in the region $A A'$ is $O(|\partial A| D)$, where $A'$ is the corresponding region in $\mathcal{H}'$. Finally, since the entanglement of formation $E^F_{AA':BB'}$ is a monotone with respective to local operations on $AA'$ and $BB'$, then
\begin{align}
    E^F_{A:B}(\rho) & = E^F_{A:B}(\Tr_{A'}\circ \Tr_{B'}(U \ketbra{\psi_0}{\psi_0} U^\dagger)) \nonumber \\ 
    & \leq E^F_{AA':BB'}(U\ketbra{\psi_0}{\psi_0} U^\dagger) \nonumber \\
    & = O(|\partial A| D), \label{eq:depth_bound}
\end{align}
implying the depth $D$ is larger than a constant multiple of the entanglement of formation. More generally, the same argument presented here can be used to prove that the entanglement of purification\footnote{``Entanglement'' here is a misnomer when compared to entanglement of formation, since $E_P$ measures both quantum and classical correlations.} $E_P$, defined as 
\begin{equation}
E_P(\rho_{AB}) = \min_{\stackrel{\ket{\psi}_{AA'BB'}}{\Tr_{A'B'} \ket{\psi} = \rho}} S_{AA'}(\ket{\psi}),
\end{equation}
is bounded from below by the entanglement of formation \cite{terhal_entanglement_2002}. 

Now, consider a circuit of unitaries and measurements. Let us note by $\vb{m}$ a possible set of joint measurement results, with $E_{\vb{m}}$ the corresponding Kraus operator that acts on the initial state $\rho_0 = \ketbra{\psi_0}{\psi_0}$ by $\rho_{\vb{m}} = p_{\vb{m}}^{-1} E_{\vb{m}} \rho_0 E_{\vb{m}}^\dagger$ when $\vb{m}$ is observed with probability $p_{\vb{m}}$. Since each $E_{\vb{m}}$ is still composed of small gates (now not necessarily unitary), then $\rho_{\vb{m}}$ is still area-law. This argument follows from noting that the growth in the R\'enyi-0 entropy of the state (which upper bounds the entanglement entropy) is obtained by singular value decomposition of the circuit across a cut, which remains area law\footnote{See Appendix D of \cite{luMeasurementShortcutLongRange2022} for a more detailed proof.}.

From this fact and the convexity of $E^{F}$, we have
\begin{align}
    E^F_{A:B}(\rho) & = E^F_{A:B}(\sum_{\vb{m}} p_{\vb{m}} \Tr_{A'}\circ \Tr_{B'}(\rho_{\vb{m}})) \nonumber \\
    & \leq \sum_{\vb{m}} p_{\vb{m}} E^F_{A:B}(\rho_{\vb{m}}) \nonumber \\
    & = O(|\partial A| D), \label{eq:depth_bound_adaptive}
\end{align}
thus bounding from below again the preparation depth $D$ by the entanglement of formation.

Since the MMIS with continuous non-Abelian symmetries have entanglement of formation scaling of $\log |A|$ by Eq.~\eqref{eq:log_EF}, then any adaptive preparation must have depth of order $O(\log N)$. We conjecture that this bound can be saturated, and we leave this for future exploration.


\subsection{Stability and strong-to-weak SSB}\label{sec:stability_swssb}

MMIS is the unique steady state of all strong-symmetric unital complete channels (as was shown in Sec.~\ref{sec:channel_thm}). This implies that it is stable to all perturbations that satisfy these conditions. However, we can even relax the unitality condition, and infer that the MMIS forms a mixed state phase which is stable to all short depth symmetric channels.

This is because the MMIS displays a new type of spontaneous symmetry breaking, which has been recently studied in the context of mixed states: strong-to-weak SSB (SW-SSB) \cite{lee2023quantum,MaetalSSB,lessaStrongtoWeakSpontaneousSymmetry2024, salaSpontaneousStrongSymmetry2024}. A state has strong-to-weak SSB if it has a strong symmetry that is broken down to a weak one, as detected by a nonlinear long-range correlator. One such correlator that is preserved against local channels, and thus defines SW-SBB for the entire symmetric mixed phase of matter, is the so-called fidelity correlator
\begin{equation}\label{eq:fidelity_def}
    F_O(i, j) \equiv F(\rho, O_i O_j^\dagger \rho O_j O_i^\dagger),
\end{equation}
where $F(\rho, \sigma) = \Tr \sqrt{\sqrt{\rho} \sigma \sqrt{\rho}}$ is the fidelity and $O_i$ is a local charged operator. In the case $O$ is part of higher dimensional representation, $U(g) O_\alpha U(g)^\dagger = \sum_\beta \mathcal{U}_{\alpha \beta} O_\beta$, then we substitute $O_i O_j$ by the symmetry-invariant quadratic operator around sites $i$ and $j$, $O^{(2)}_{i,j} \equiv \sum_{\alpha} O_{i,\alpha} O^\dagger_{j, \alpha}$. Then, by Schur's lemma, $O^{(2)}_{i,j}$ acts only on the multiplicity space of each irrep, allowing us to calculate the fidelity correlator of any maximally mixed state $\rho_J$ exactly:
\begin{align}\label{eq:fidelity_corr_mms}
    F_O(i, j) & = \Tr \sqrt{\sqrt{\rho_J} O^{(2)}_{i,j} \rho_J (O^{(2)}_{i,j})^\dagger \sqrt{\rho_J}} \nn \\
    & = \frac{1}{d_J} \Tr[ P_J \sqrt{O^{(2)}_{i,j} (O^{(2)}_{i,j})^\dagger}] \nn \\
    & = \Tr[\rho_J\ \left|O^{(2)}_{i,j}\right|],
\end{align}
which is strictly positive as $O^{(2)}_{i,j} \neq 0$ 
, and is independent of the sites $i \neq j$ due to the permutation symmetry of $\rho_J$ (See Lemma \ref{lemma:permutation_invariance}).

An instance of such symmetry-invariant quadratic operator which allows us to calculate the fidelity correlator exactly is the operator constructed out of generators of Lie algebra, as in Eq.~\eqref{eq:casimir_correlator}. Given $(X_{i,\alpha})_{\alpha=1}^{\dim \mathfrak{g}}$, an orthonormal basis of the Lie algebra $\mathfrak{g}$ of $G$ with respect to the Killing form, for each local qudit $i$, we consider the quadratic operator,
\begin{align}
    O^{(2)}_{i,j} \equiv \sum_{\alpha}U(X_{i,\alpha})U(X^{\dagger}_{j,\alpha}).
\end{align}
Let $c_{J}$ the Casimir eigenvalue corresponding to the irrep $J$. It can be easily confirmed that $O^{(2)}_{i,i} + O^{(2)}_{i,j} + O^{(2)}_{j,i}+ O^{(2)}_{j,j}$ is the representation of the quadratic Casimir element on the 2 qudits $i,j$. By decomposing it in terms of projectors on different irreps, we get the relation,
\begin{align}\label{eq:fidelity_quadratic_op}
    O^{(2)}_{i,j} = \sum_{J}\left(\frac{c_{J}-2c_{\mathcal{V}}}{2}\right)P_{J}^{(i,j)},
\end{align}
where $c_{\mathcal{V}}$ is the Casimir eigenvalue for the irrep defining the local qudit, and $P_{J}^{(i,j)}$ is the projector of the Hilbert space of $i,j$ qudits on to the symmetry sector $J$. This relation, when restricted to $SU(2)$ and qubits, is the familiar relation that the Heisenberg term of two spin-$\frac{1}{2}$ decomposes as, $\vb{S}_i \cdot \vb{S}_j = -3/4 P_{0} +1/4 P_{1} $.

Now, we can compute the fidelity correlator of the MMIS using Eq.~\eqref{eq:fidelity_corr_mms}, and the bipartite decomposition of $\rho_\triv^{N}$ in Eq.~\eqref{eq:mms_decomp}. The relevant bipartite decomposition of $N$ qudits is: $(i,j):(i,j)^c$, where $(i,j)$ refer to the $i,j$ qudits in the definition of the fidelity correlator. Eq.~\eqref{eq:fidelity_corr_mms} evaluated on MMIS and with the quadratic invariant operator Eq.~\eqref{eq:fidelity_quadratic_op} is given by,
\begin{align}
    F_O(i, j)\left(\rho^{N}_\triv\right) = \sum_{J \in \mathcal{J}_{(i,j)(i,j)^{c}}}\left|\frac{c_J-2c_{\mathcal{V}}}{2}\right|\frac{\mathfrak{C}^{2}_{J}\mathfrak{C}^{N-2}_{J^{*}}}{\mathfrak{C}_\triv^{N}}.
\end{align}
The sum is over finitely many $O(1)$ numbers (recall the normalization relation, Eq.~\eqref{eq:multiplicity_factorization}), and is thus generically $O(1)$ even at the thermodynamic limit $N\to \infty$. 

For $SU(2)$, the multiplicities can be computed exactly using Eq.~\eqref{eq:su2_multiplicity}, and the Casimir eigenvalues are given by $c_{j} = j(j+1)$ for spin$-j$ irreps. The fidelity correlator in the thermodynamic limit can be computed using these formulas,
\begin{align}
F^{SU(2)}_O(i, j)\left(\rho^{N}_\triv\right) \overset{N\to\infty}{\longrightarrow} \frac{3}{8},  
\end{align}
while the spin-spin correlation decays algebraically as $1/N$ with $N \gg 1$, as shown in Eq.~\eqref{eq:su2_correlation}. This behavior of constant fidelity correlator and algebraically decaying linear correlator persists for all continuous non-Abelian Lie groups. This implies that such states exhibit SW-SSB in the thermodynamic limit.


For discrete non-Abelian groups, we can exactly compute the linear and quadratic\footnote{Long-range quadratic correlators also probe SW-SSB, similarly to the fidelity of Eq. \eqref{eq:fidelity_def}. See \cite{lessaStrongtoWeakSpontaneousSymmetry2024, salaSpontaneousStrongSymmetry2024} for more details.} correlators for the MMIS for some specific examples. In Sec. \ref{appsec:s3_correlations}, we show two examples of quadratic operators for the $S_3$ group with constant fidelity correlator but linear correlator that is exponentially small in system size. Hence, the MMIS of $S_3$ has SW-SSB. We expect this to happen to all non-Abelian discrete groups. 

These results give further evidence to the conjecture made in \cite{lessaStrongtoWeakSpontaneousSymmetry2024} that all thermal states in the canonical ensemble -- having strong symmetry -- at nonzero temperature exhibit SW-SSB if they don't have conventional SSB already. Our calculations confirm this hypothesis in the infinite temperature limit and it would also confirm for high enough temperature states if they are in the same mixed-state phase as the infinite temperature state.

In~\cite{lessaStrongtoWeakSpontaneousSymmetry2024} it was further proven that states with SW-SSB cannot be brought to a symmetric pure product state using a symmetric low depth $O(PolyLog(N))$ circuit. This proves that the SW-SSB nature of the MMIS is stable to symmetric low-depth perturbations. However, this argument does not imply that the entanglement structure is also stable to such perturbations.

\section{Concluding Remarks}
In this work we have identified strong non-Abelian symmetry of quantum channels as a sufficient condition for long-range entanglement (with logarithmic scaling with subsystem size) in mixed steady states of such channels. 

There are several directions to consider in the future. We have established that such steady states are unique for symmetric unital channels and the strong-weak SSB is stable to symmetric short depth channels. However, symmetric non-unital channels can lead to distinct steady states. It is not apriori clear if local symmetric non-unital channels can lead to steady states without such long-range entanglement. That the projector onto the symmetric space for semisimple Lie groups is highly entangled is reminiscent of the fact that the singlets under such groups provide a decoherence free subspace for symmetric noise~\cite{Lidar_2003}, which as a quantum error correcting code has infinite distance (for such restricted symmetric noise). This observation may be useful to prove that the steady states with local non-unital symmetric perturbation may also be long-range and entangled. In this work we have also focused on the maximally mixed invariant steady state arising due to strong on-site symmetry of the channel. However, interesting non-stationary behavior may arise due dynamical symmetries which lead to persistent oscillations due to eigenmodes of the channel with purely imaginary eigenvalues~\cite{Buca_2019, Buca2022}. It is an interesting open direction to study whether under non-Abelian continuous dynamical symmetries the mixed states also undergo non-stationary entanglement behavior.

We also proved a lower bound of $\log N$ depth to adaptively prepare the maximally mixed invariant state for the case of continuous semisimple non-Abelian symmetries. It is not clear whether such a bound may be saturated. 

To prove the equality of the entanglement measures of formation and distillation for the MMIS and to quantify them, we have relied on the property of local distinguishability via LOCC, as previous works on the $SU(2)$ symmetry have as well \cite{livine_entanglement_2005, livine_quantum_2006}.  Although symmetry sectors other than the invariant one may not share this feature (See Sec. \ref{sec:other_sym_sectors}), it would be interesting to study other states with the same property, irrespective of symmetry, particularly because of its extensibility to mixed many-body systems. 

Another interesting perspective comes from considering strongly symmetric thermal states. In particular, the MMIS of $SU(2)$ can be considered as the infinite temperature canonical ensemble state, where the bath exchanges energy but no $SU(2)$ charge. Our result shows that such a thermal state is highly entangled at infinite temperature, which is in high contrast with the usual (grand canonical) Gibbs ensemble $e^{-\beta H}/\mathcal{Z}$, where the strong symmetry is explicitly broken. Indeed, it has recently been shown that Gibbs states of local Hamiltonians at high enough temperatures $T > T_c$ are fully separable, efficiently preparable, and easy to sample~\cite{bakshi_high-temperature_2024}. Hence, one possible use of strongly symmetric thermal states would be as resources for quantum advantage experiments. They could be preparable on a given quantum platform if, for example, the interaction with an environment at finite temperature is naturally constrained by a symmetry. Note, however, that further analysis on the sampling complexity for such states is needed. Finally, symmetry-enforced entanglement could imply a mixed-state phase transition going from low to high temperatures, due to possibly distinct scaling patterns of entanglement.

It is also interesting to compare the behavior of logarithmic negativity and entanglement of formation in such symmetric mixed states. Generally, quantities such as entanglement of formation are extremely hard to compute for many-body systems; however, the MMIS of generic symmetries is a class of non-trivial states for which the entanglement of formation can be exactly computed. In the examples considered here, we find that the bipartite logarithmic negativity and entanglement of formation match in the thermodynamic limit for large enough bipartitions; however it will be interesting to find cases where they are parametrically different. The maximally mixed invariant states provide a simple non-trivial class of states to compare the different measures of mixed state entanglement.

\textbf{Acknowledgement: }SS thanks Arghya Sadhukhan, Tarun Grover, and Pablo Sala for useful discussions, and Aspen Center for Physics, supported by National Science Foundation grant PHY-2210452, where the project was initiated. LAL acknowledges supports from the Natural Sciences and Engineering Research Council of Canada (NSERC) through Discovery Grants. AM thanks Thomas Scaffidi for insightful discussions. Research at PI is supported in part by the Government of Canada through the Department of Innovation, Science and Economic Development Canada and by the Province of Ontario through the Ministry of Colleges and Universities.

\textit{Note added: }While this work was being completed, we became aware of a related work by Li, Pollmann, Read, and Sala~\cite{li_highly2024}, which will appear in the same arXiv posting.

\appendix

\section{Proofs for Sec.~\ref{sec:exp_values}} \label{appsec:exp_values_proof}


In this section, we prove the statement that in the thermodynamic limit $N\to \infty$ and a constant $k \sim O(1)$ body reduced density matrix, $J$ irreps appear in $\mathcal{V}^N$ with multiplicities $0$ or  
\begin{equation}\label{eq:appendix_mults}
    \mathfrak{C}^{N}_{J} = f_{J}d_{\mathcal{V}}^N \times %
    \begin{cases}
        (1+O(e^{-\alpha N})), & \text{(discrete non-Abelian)} \\
        \frac{1}{N^\beta}(1+O(N^{-1})), & \text{(semisimple Lie group)} \\ 
    \end{cases}
\end{equation}
for $f_J$ a coefficient, $\alpha > 0$ and $\beta = \dim G / 2$.

Using this result, we sketch a proof for Theorem~\ref{thm:local_indist} in Sec.~\ref{sec:exp_values}.

\subsection{Discrete non-Abelian symmetries}

Before we prove \ref{eq:appendix_mults} for finite groups, we will first prove a lemma regarding the maxima of group characters:
\begin{lemma}\label{lemma:center_faithful}
    Given an irrep $\mathcal{V}$,
    \begin{equation}
        g \in Z(G) \Rightarrow |\chi_\mathcal{V}(g)| = d_\mathcal{V},
    \end{equation}
    and the converse is also true if $\mathcal{V}$ is faithful.
\end{lemma}
\begin{proof}
    ($\Rightarrow$) If $g \in Z(G)$, then $U_\mathcal{V}(g)$ is an intertwiner of $\mathcal{V}$. Thus, by Schur's lemma, $U_\mathcal{V}(g) = \lambda_g \mathbb{I}$ and $|\chi_V(g)| = d_\mathcal{V} |\lambda_g| = d_\mathcal{V}$.

    ($\Leftarrow$) If $|\chi_V(g)| = d_\mathcal{V}$, then $U_\mathcal{V}(g) = \lambda_g \mathbb{I}$ must be a phase multiple of the identity, since $\chi_V(g)$ is the sum of the $d_\mathcal{V}$ eigenvalues of $U_\mathcal{V}(g)$, all of which reside in the unit circle. Then, for any $h \in G$, 
    \begin{equation}
        U_\mathcal{V}(g^{-1} h g h^{-1}) = \lambda_g^{-1} U_\mathcal{V}(h) (\lambda_g \mathbb{I}) U_\mathcal{V}(h)^{-1} = \mathbb{I},
    \end{equation}
    which implies $g h = h g$ from the faithfulness of $\mathcal{V}$. Since $h \in G$ was arbitrary, $g \in Z(G)$.
\end{proof}

\begin{lemma}\label{lemma:asymptotic_multiplicities}
    In the limit $N \to \infty$ and for finite groups, the multiplicity $\mathfrak{C}_J^N$ of an irrep $J$ in the tensor power $\mathcal{V}^{\otimes N}$ of a faithful representation $\mathcal{V}$ satisfies
    \begin{equation}\label{eq:asymptotic_formula}
        \mathfrak{C}_J^N = 0 \text{ or } f_J d_\mathcal{V}^N (1 + O(e^{-\alpha N})),
    \end{equation}
    where $f_J = d_J |Z(G)| / |G| > 0$, $\alpha > 0$ depends only on $\mathcal{V}$, and $Z(G)$ is the center of $G$.
\end{lemma}
\begin{proof}
    From Schur orthogonality relations, we can express the multiplicity $\mathfrak{C}_J^N$ as
    \begin{align}
        \mathfrak{C}_J^N & = \langle \chi_J, \chi_\mathcal{V}^N \rangle_G, \nonumber \\
        & = \frac{1}{|G|} \sum_{g \in G} \overline{\chi_J(g)} \chi_\mathcal{V}(g)^N.
    \end{align}
    The largest absolute value of $\{ \chi_\mathcal{V}(g) \mid g \in G\}$ is $d_\mathcal{V}$, saturated precisely for $g \in Z(G)$, due to Lemma \ref{lemma:center_faithful}. Hence,
    \begin{equation}
        \mathfrak{C}_J^N = \frac{1}{|G|} \sum_{g \in Z(G)} \overline{\chi_J(g)} \chi_\mathcal{V}(g)^N + O(\lambda_2^N),
    \end{equation}
    where $\lambda_2 < d_\mathcal{V}$ is the second largest absolute value. We can calculate the dominant part above by noting that reducing any representation $V$ on $G$ to $Z(G)$ gives $d_{\mathcal{V}}$ copies of a 1d irrep $V_Z$, formed by the phase factors $\lambda_g$ of $U_V(g \in Z(G)) = \lambda_g \mathbb{I}$. Thus, 
    \begin{align}
        \sum_{g \in Z(G)} \overline{\chi_J(g)} \chi_\mathcal{V}(g)^N & = d_J d_\mathcal{V}^N \sum_{g \in Z(G)} \overline{\chi_{J_Z}(g)} \chi_{\mathcal{V}_Z}^N(g) \\
        & = d_J d_\mathcal{V}^N |Z(G)| \langle \chi_{J_Z}, \chi_{\mathcal{V}_Z}^N \rangle_{Z(G)}.
    \end{align}
    From Schur orthogonality relations on $Z(G)$, the inner product above is either zero or equal to 1. Moreover, it is zero only if $\mathfrak{C}_J^N = 0$, otherwise $J \in \mathcal{V}^{\otimes N}$ would imply $J_Z \in \mathcal{V}_Z^{\otimes N}$. Therefore, we arrive at Eq.~\eqref{eq:asymptotic_formula} with $f_J = d_J |Z(G)| / |G|$ and $\alpha = \log(d_\mathcal{V} / \lambda_2) > 0$.
\end{proof}

Another lemma is necessary for the occurrence of irreps in tensor powers of another irrep:
\begin{lemma}\label{lemma:faithful_everywhere}
    For $\mathcal{V}$ a faithful representation of a finite group $G$ and $J$ an irrep, $J$ appears in the tensor powers $\mathcal{V}^{\otimes N_i}$, $i \geq 0$, for $N_i \to \infty$ an increasing sequence.
\end{lemma}
\begin{proof}
    Due to Burnside-Brauer theorem \cite{burnsideTheoryGroupsFinite1955, brauerNoteTheoremsBurnside1964, fulton_representation_2004}, any irrep is contained in some tensor power of a faithful representation $\mathcal{V}$. Let $N_\triv$ and $N_J$ be the first time the trivial and the $J$ irrep appear in $\mathcal{V}^{\otimes N}$, respectively. Hence, since $\triv \otimes J \cong J$, then $J$ appears in all $\mathcal{V}^{\otimes N_i}$ powers, with $N_i = N_J + i N_\triv$, $i \geq 0$.
\end{proof}

Given the lemmas above, the proof for finite groups goes as follows. From Eq.~\eqref{eq:reduced_MMIS} for the reduced density matrix, we have
\begin{equation}\label{eq:k-site_reduced}
    \rho^{N,k}_\triv = \sum_{J\in \mathcal{J}_{k,N-k}} \sum_{a = 1}^{\mathfrak{C}^{k}_{J}}\frac{\mathfrak{C}^{N-k}_{J^*}}{\mathfrak{C}_{\triv}^{N}d_{J}}\mathbb{I}_{J}\otimes \ket{a}\bra{a},
\end{equation}
where $\mathcal{J}_{k,N-k}$ is the set of irreps $J$ that appear in $\mathcal{V}^{\otimes k}$ and whose conjugate $J^*$ appear in $\mathcal{V}^{\otimes(N-k)}$. In turn, the maximally mixed state can be decomposed in the same basis as
\begin{equation}\label{eq:k-site_MM}
    \frac{1}{d_{\mathcal{V}}^k} \mathbb{I}_{\mathcal{V}^{\otimes k}} = \sum_{J\in \mathcal{J}_k} \sum_{a = 1}^{\mathfrak{C}^{k}_{J}}\frac{1}{d_{\mathcal{V}}^k}\mathbb{I}_{J}\otimes \ket{a}\bra{a},
\end{equation}
where $\mathcal{J}_{k}$ is the set of irreps $J$ that appear in $\mathcal{V}^{\otimes k}$. For high enough $N$ and assuming a faithful representation $\mathcal{V}$ of $G$, $\mathcal{J}_{k,N-k} = \mathcal{J}_k$ for at least an infinite increasing sequence $N_i \to \infty$ of system sizes, from Lemma \ref{lemma:faithful_everywhere} applied to $\mathcal{V}^{\otimes (N-k)}$. Given an $N$ in this sequence, we can use Eq.~\eqref{eq:asymptotic_formula} of Lemma \ref{lemma:asymptotic_multiplicities} to simplify the components of $\rho_\triv^{N,k}$ to 
\begin{equation}
    \frac{\mathfrak{C}^{N-k}_{J^*}}{\mathfrak{C}_{\triv}^{N}d_{J}} = \frac{1}{d_\mathcal{V}^k} [1 + O(e^{-\alpha (N-k)})],
\end{equation}
where we have used that $f_{J^*} = f_\triv d_J$. Thus, the trace distance between $\rho^{N,k}_\triv$ and $\frac{1}{d_{\mathcal{V}}^k} \mathbb{I}_{\mathcal{V}^{\otimes k}}$ is exponentially small in $N-k$: 
\begin{align}
    D(\rho^{N,k}_\triv, d_{\mathcal{V}}^{-k} \mathbb{I}_{\mathcal{V}^{\otimes k}}) & = \frac{1}{2} \sum_{J\in \mathcal{J}_k} \sum_{a = 1}^{\mathfrak{C}^{k}_{J}} \left| \frac{\mathfrak{C}^{N-k}_{J^*}}{\mathfrak{C}_{\triv}^{N}d_{J}} - \frac{1}{d_{\mathcal{V}}^k}\right| \nonumber \\
    & = \frac{1}{2} \sum_{J\in \mathcal{J}_k} \sum_{a = 1}^{\mathfrak{C}^{k}_{J}} \frac{1}{d_{\mathcal{V}}^k} O(e^{-\alpha(N-k)}) \nonumber \\
    & = O(e^{-\alpha(N-k)}).
\end{align}

By using the fact that $D(\rho, \sigma) = \max_{0 \leq P \leq \mathbb{I}} \Tr[P(\rho - \sigma)]$, we get as a corollary that the expectation values of the MMIS and the identity state on $k$ qubits are also exponentially close (See Eq. \eqref{eq:exp_vals_corollary}).
 
\subsection{Semisimple Lie groups}
The multiplicities of semisimple Lie groups also satisfy the Schur orthogonality relation,
\begin{align}\label{eq:schur_lie}
        \mathfrak{C}_J^N & = \langle \chi_J, \chi_\mathcal{V}^N \rangle_G, \nonumber \\
        & = \int_{G} dg \overline{\chi_J(g)} \chi_\mathcal{V}(g)^N,
    \end{align}
where $dg$ is the Haar invariant measure of $G$.
    
Let us assume the case of trivial center $Z(G)$ for simplicity\footnote{For groups with nontrivial discrete center, there will only be a change in the proportionality constant of $\mathfrak{C}_J^N$.}. We can restrict the integral above to a neighbourhood around the identity element, since there the characters peak at their global maximum $\chi_j(e) = d_j$. Assume a quadratic expansion around the identity, $\chi_{J}(e + \epsilon t) \approx d_J(1-a_J \epsilon^{2})$, $\chi_{\mathcal{V}}(e +\epsilon t) \approx d_{\mathcal{V}}(1-a_{\mathcal{V}} \epsilon^{2})$, for $t \in \mathfrak{g}$ and small $\epsilon$.\footnote{Dependencies of $a_j$ with $t$ will not affect the end results, so we assume isotropy for simplicity.} In the thermodynamic limit, $\chi_{\mathcal{V}}(e +\epsilon t)^{N}$ can be approximated by a Gaussian around $\epsilon = 0$, of width $1/\sqrt{a_\mathcal{V} N}$. Thus, a good estimate of Eq.~\eqref{eq:schur_lie} is obtained from integrating around a $m = \dim \mathfrak{g}$ dimensional ball of radius $1/\sqrt{a_\mathcal{V} N}$ around $e$, 
\begin{equation}
    \mathfrak{C}_J^N \approx d_J d_{\mathcal{V}}^{N} \int_{S^{m-1}} d\Omega \int_{0}^{\frac{1}{\sqrt{a_\mathcal{V} N}}} d \epsilon \left[e^{-N \epsilon^{2}/a_\mathcal{V}}(1- a_J \epsilon ^{2})\right], 
\end{equation}
This formula can be evaluated to give
\begin{equation}\label{eq:lie_mult}
    \mathfrak{C}_J^N \approx c \frac{d_J d_{\mathcal{V}}^{N}}{N^{\dim \mathfrak{g}/2}}(1 - O(N^{-1})),
\end{equation}
for some constant $c$. This leads to the following result,
\begin{equation}
    \frac{\mathfrak{C}^{N-k}_{J^*}}{\mathfrak{C}_{\triv}^{N}d_{J}} = \frac{1}{d_\mathcal{V}^k} [1 + O(k/N)],
\end{equation}
This shows, by repeating the same arguments as the discrete non-Abelian case, that the trace distance between $\rho^{N,k}_\triv$ and identity state is algebraically small for semisimple Lie groups. 

Note, Eq.~\eqref{eq:lie_mult} is true in general for all $J$, but not very useful for finding the peak of the distribution of the multiplicities, which is needed for the asymptotic computation of half chain entanglement of formation (Sec.~\ref{sec:lie_mmis_gen}).

\section{$S_3$ maximally mixed invariant state}\label{appsec:s3_details}

\subsection{Correlators in the $S_3$ MMIS}\label{appsec:s3_correlations}

We first find an expression for the $S_3$ MMIS in terms of the Pauli operators. This helps us calculate the correlations explicitly. First of all, we note that for any unitary representation of the $S_3$ group, the followings are projection operators onto different sectors of the $\mathbb{Z}_3$ rotation subgroup:
\begin{eqnarray}
  P_0 = \frac{\mathbb{I} + R + R^{-1}}{3}, \quad  P_{\theta} = \frac{\mathbb{I} + e^{- i \theta} R + e^{ i \theta} R^{-1}}{3},  
\end{eqnarray}
where $P_0$ is the projector onto the singlet sector of the rotation (i.e., $R P_0 = P_0$), and $P_{\theta}$ is the projection onto $\{\ket{\theta}\}$ subspace of the $\twod$ representation (i.e. $R P_{\theta} = e^{i \theta} P_{\theta}$, and $P_x P_{\theta} P_x = P_{-\theta}$). Therefore, the projector onto the whole $\twod$ representation is
\begin{eqnarray}
   P_{\twod} = P_{2\pi/3} + P_{-2\pi/3} = \mathbb{I} - P_0. 
\end{eqnarray}
In order to find the projectors onto the trivial and sign sectors, we need to project $P_0$ further onto even and odd sectors of the reflection:
\begin{align}
    P_{\triv} &= \frac{\mathbb{I} + R + R^{-1}}{3} \frac{\mathbb{I} + P_x}{2}, \nonumber \\
    P_{\sgn} &= \frac{\mathbb{I} + R + R^{-1}}{3} \frac{\mathbb{I} - P_x}{2}
\end{align}
It is not difficult to confirm that both $P_{\triv}$ and $P_{\sgn}$ are invariant under rotation but are even and odd under reflection respectively (i.e., $R P_{\triv} = P_{\triv}$, $P_x P_{\triv} = P_{\triv}$ and $R P_{\sgn} = P_{\sgn}, P_x P_{\sgn} = -P_{\sgn}$). Now, for an $N$-particle system, we can take $R = \prod e^{\frac{2 \pi i}{3} Z_j}$ and $P_x = \prod X_j$. The MMIS is the $P_{\triv}$ projector while its trace is normalized to 1:
\begin{eqnarray}
    \rho_{\triv} = \frac{1}{\mathfrak{C}_{\triv}^N} \frac{\mathbb{I} + \prod e^{\frac{2\pi i}{3} Z_j} + \prod e^{-\frac{2\pi i}{3} Z_j}}{3} \frac{\mathbb{I} + \prod X_j}{2}
\end{eqnarray}

In order to capture strong/weak symmetry breaking, we should seek operators with charged one-point functions and invariant two-point functions. With this choice, the one-point functions always vanish due to the strong symmetry of the MMS in the singlet sector of the group. For the case of $S_3$, $\expval{Z_i Z_j}$ and $\expval{X_i X_j + Y_i Y_j}$ are the simplest correlations satisfying the above conditions. Using the expansion $\exp{\pm \frac{ 2\pi i}{3} Z_j} = -\frac{1}{2} \mathbb{I} \pm i \frac{\sqrt{3}}{2} Z_j$ and the fact that any Pauli operator is traceless, one can calculate the following linear and quadratic correlators (in the large $N$ limit):
\begin{align}\label{eq:lin_corr_S3}
    \Tr[\rho_{\triv} Z_i Z_j] &= \frac{6 (-1)^{N+1}}{2^N} \nonumber \\
    \Tr[ \rho_{\triv} (X_i X_j + Y_i Y_j)] &= 0
\end{align} 

\begin{align}
    &\frac{\Tr\big[(\rho_{\triv} Z_i Z_j)^2\big]}{\Tr \rho_{\triv}^2} = 1, \nonumber \\
    &\frac{\Tr\big[(\rho_{\triv} (X_i X_j + Y_i Y_j))^2\big]}{\Tr \rho_{\triv}^2} = 2 - 12 \frac{(-1)^{N+1}}{2^N}
\end{align}
Eq.~\eqref{eq:lin_corr_S3} shows that the linear correlations are exponentially (super-exponentially for $XX+YY$) short-ranged. 

Using the results from Appendix \ref{appsec:exp_values_proof},
we can prove something even stronger, that the mutual information $I(i,j) \equiv S(\rho_\triv^i)+S(\rho_\triv^j)-S(\rho_\triv^{ij})$ is also $O(e^{-\alpha N})$ for any single qubit sites $i,j$ which implies that all connected correlations are exponentially suppressed~\cite{Wolf_2008} in system size $N$. 


This behavior as well as the $O(1)$ scaling of the entanglement suggests that the mixed state $\rho_{\triv}$ is likely to be SRE, although they cannot rule out the long-range entanglement possibility.
In the following sections we discuss a method letting us determine numerically whether the state is SRE or LRE. The only caveat will be then how large the finite size effects might be.

\subsection{A local $S_3$-symmetric channel}\label{appsec:s3_channel}

Here we give a concrete construction of a strongly $S_3$-symmetric local measurement channel.  First, we note that there are no on-site operators invariant under $S_3$. Therefore we need to involve higher-body measurements. The simplest channel which involves 2-body terms is to measure the irrep of nearest neighbor pairs in a brick wall structure. The corresponding  local gate acting on sites $i,j$ is thus given by
\begin{align}
    \mathcal{E}^{ij}_{\text{2-body}}(\rho) = K_{\triv}^{ij} \rho K_{\triv}^{ij} + K_{\sgn}^{ij} \rho K_{\sgn}^{ij} + K_{\twod}^{ij} \rho K_{\twod}^{ij},
\end{align}
where the local Kraus operators are
\begin{align}
    & K_{\triv}^{ij} = \frac{\mathbb{I}-Z_i Z_j}{2} \frac{\mathbb{I}+X_i X_j}{2} \nonumber \\ & K_{\sgn}^{ij} = \frac{\mathbb{I}-Z_i Z_j}{2} \frac{\mathbb{I}-X_i X_j}{2} \nonumber \\ & K_{\twod}^{ij} = \frac{\mathbb{I}+Z_i Z_j}{2} 
\end{align}
This channel looks very simple and generates a Clifford dynamics \footnote{The dynamics can be interpreted in this way: at each step, choose a bond at random and measure $Z_iZ_j$. If the outcome is $+1$, do nothing more; otherwise, measure $X_i X_j$ of the same bond.} letting us efficiently simulate the channel in a polynomial time. However, the channel has an additional $\mathbb{Z}_2$ symmetry generated by $P_z = \prod_{\text{all sites}} Z_i$. The generators of local symmetry group can be chosen as $\{e^{2\pi i/3 Z},X,Z\}$ which generate a $S_3 \rtimes \mathbb{Z}_2$ group. The additional $\mathbb{Z}_2$ subgroup can lead to a steady state degeneracy in charge sectors of the $S_3$ subgroup. To break the additional $\mathbb{Z}_2$ symmetry, one should involve 3-body measurements. The simplest realization of this is to measure the irrep of three neighboring sites:
\begin{align}
    \mathcal{E}^{ijk}_{\text{3-body}}(\rho) =  L_{\triv}^{ijk} \rho L_{\triv}^{ijk} +& L_{\sgn}^{ijk} \rho L_{\sgn}^{ijk} + L_{\twod}^{ijk} \rho L_{\twod}^{ijk},
\end{align} 
where $L^{ijk}$ operators are projectors onto different $S_3$ sectors of the sites $i,j,k$. This channel has also undesired symmetries. Therefore, in order to construct a purely $S_3$ symmetric channel, one can combine both types of gates in a brick wall structure depicted in Fig.~\ref{fig:S3_neg_vs_time}. The channel is both complete and unital, hence by the results of Sec.~\ref{sec:channel_thm}, the MMIS is the unique steady state in the invariant sector.

\subsection{$S_3$ MMIS: SRE or LRE?}

\begin{figure}[t]
\centering
  \includegraphics[height=6.cm]{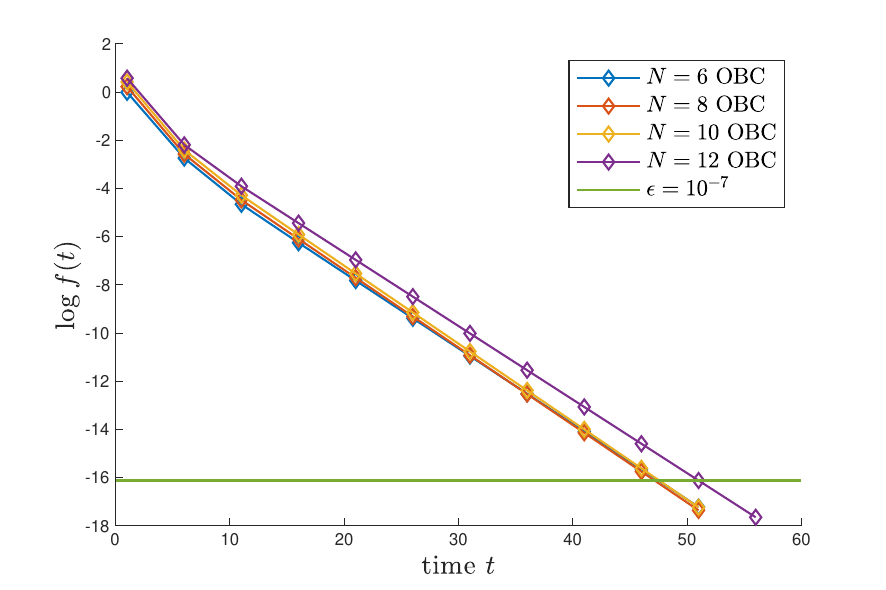}
  \caption{$\log f(t)$ vs time, for the random unitary channel with the parameters $q = 0.5$, $\phi_1 = \pi$ and $\phi_2 = \pi/2$. Here $f(t)$ shows how fast the density matrix is converging to the steady state and is defined as $f(t) := ||\rho_t - \rho_{t-1}||_1/\sqrt{||\rho_t||_1 \hspace{3pt} ||\rho_{t-1}}||_1$. where $\rho_t$ is the density matrix after depth $t$. The system is initialized with nearest neighbor $\triv$ pairs.}
  \label{fig:S3_conv_time_unit}
\end{figure}

One can numerically confirm that the $S_3$ MMIS can be generated at finite depth from a product state through the application of the local channel introduced in the last section. This, however, does not necessarily imply that the state is SRE in the sense that it can be decomposed into an ensemble of SRE pure states. One can now ask the question that under which conditions the finite convergence time implies the SREness of the state. To answer this, we note that the trajectories of a channel form a natural pure state decomposition for the steady state. Therefore, one sufficient condition is that the pure state trajectories of the channel undergo a finite depth unitary circuit. Fortunately, generic strongly symmetric local channels are complete and lead to the same unique state steady as proven in Sec.~\ref{sec:channel_thm}. Therefore, instead of the channel used earlier which involve only measurement gates, one can replace the local measurement gates with a classical combination of identity gates and symmetric unitaries. Here is a systematic construction of such channel: assume the projectors corresponding to local measurement outcomes form a set $\{P_i\}$. We define a local unitary operator by $U = \sum_{i} e^{i\phi_i} P_i$, where $\phi_i$'s are some parameters. One should then replace any measurement with the local channel $\mathcal{E}_u(\rho) = (1-q) \rho + q U \rho U^{\dagger} $. Each trajectory of this channel realize a unitary evolution as desired. For the specific case of $S_3$ symmetry we choose, $q = 0.5$ and $U = K_{\triv} + e^{i \phi_1} K_{\sgn} + e^{i \phi_2} K_{2d}$ (for both 2-body and 3-body operations). Fig.~\ref{fig:S3_conv_time_unit} shows that for an $S_3$-symmetric channel constructed in this way, the convergence time does not scale with the system size, suggesting that the $S_3$ MMIS must be SRE.

\section{$SU(d)$ maximally mixed invariant state}\label{appsec:su(d)}

Consider a chain of qudits, with the local qudit space being a $d$ dimensional complex Hilbert space that naturally carries the fundamental irrep of $SU(d)$ symmetry. We can use Eq.~\eqref{eq:mmis_entanglement} to characterise the bipartie entanglement of the maximally mixed $SU(d)$ invariant state on this chain for any bipartition. For $SU(d)$, this computation can be carried out exactly using inputs from Schur-Weyl duality, which allows one to enumerate the multiplicities of the $SU(d)$ irreps in the tensor product space.

\subsection{Schur-Weyl Duality and $SU(d)$ symmetry}

Let $V_d$ be a d-dimensional qudit Hilbert
space that carries the fundamental irrep of $SU(d)$ symmetry. The tensor product space  of $N$ qudits $V_{d}^{\otimes N}$ can be block decomposed into irreps of the global $SU(d)$. Schur-Weyl duality~\cite{goodman_symmetry_2009} states that the $SU(d)$ group and the permutation group on $N$ objects, $S_N$, jointy decompose $V_{d}^{\otimes N}$ into irreps of the respective groups as,
\begin{align}\label{eq:sw_duality}
    V_{d}^{\otimes N} = \bigoplus_{\lambda \in \text{Par}(N, d)} V_{\lambda} \otimes S_{\lambda},
\end{align}
where $V_{\lambda}$ are the irreps of $SU(d)$ on $V_{d}^{\otimes N}$, and $S_{\lambda}$ describe the irreps of $S_N$. $\text{Par}(N,d)$ refers to the set of partitions of $N$ into $d$ parts, where each such partition is described by sequence of $d$ positive integers $(\lambda_1, \lambda_2,\cdots,\lambda_{d})$, with $\lambda_{i+1}>\lambda_{i}$, and $\sum_{i}\lambda_i = N$. The partitions can be described conveniently by a Young diagram (here the rows are numbered $0,\dots, d-1$ from the bottom),
\newline
\begin{center}
\begin{tabular}{r@{}l}
\raisebox{-6ex}{$d\left\{\vphantom{\begin{array}{c}~\\[12ex] ~
\end{array}}\right.$} &
\begin{ytableau}
~       &        &       &       &      &   \none[\dots]        \\
~       &        &       &        &   \none[\dots]  & \none  \\
\none[\vdots]  & \none & \none & \none & \none         & \none \\
~      &         &       & \none & \none         & \none \\
~      &         & \none & \none & \none  &  \none      
\end{ytableau}\\[-1.5ex]
&$\underbrace{\hspace{3em}}_{\displaystyle a}\underbrace{\hspace{0.5em}}_{\displaystyle p_{1}}\underbrace{\hspace{3em}}_{\displaystyle p_{2}} \dots    \underbrace{\hspace{1em}}_{\displaystyle p_{d-1}}$
\end{tabular}    
\end{center}
where $N$ symbols are divided into $d$ rows, specified by $d-1$ non-negative integers $(p_{1},p_{2}, \cdots, p_{d-1})$. The length of the $0$th row, $a$, is determined by setting $a = N-\sum_{i = 1}^{d-1}(d-i)p_{i}$. The subsequent rows, for $i = 1,\dots, d-1$ have lengths $\lambda_i = a+\sum_{s = 1}^{i}p_{s}$.

Each Young diagram $\lambda$ (or equivalently, partition $\in \text{Par}(N,d)$) describes an irrep of $S_{N}$. The states within the irrep can be labeled by the inequivalent Young tableaux for each Young diagram. By the Schur-Weyl duality and the representation theory of $SU(d)$, the Young diagrams also describe irreps of $SU(d)$. Eq.~\eqref{eq:sw_duality} provides a convenient description of the multiplicities of $V_{\lambda}$ in $V_{d}^{\otimes N}$: it is given by the dimension of the $S_{\lambda}$ irrep of the permutation group, $\mathfrak{C}_{\lambda} = \text{dim}(S_{\lambda})$. The dimension of the $SU(d)$ irrep is denoted by $d_{\lambda} = \text{dim}(V_{\lambda})$.

The dimension of the $S_{N}$ irreps are given by the `hook length formula'~\cite{fulton_representation_2004}, 
\begin{align}
    \mathfrak{C}^{N}_{\lambda} = \text{dim} S_{\lambda} \equiv \frac{N!}{\prod_{(x,y)\in \lambda}h_{x,y}}
\end{align}
where $(x, y)$ specifies a box from the Young diagram $\lambda$ by its row and column numbers, and $h(x, y)$ counts the number of all boxes to the right of and below $(x, y)$, including itself. The dimension of the $SU(d)$ irreps are given by the `factor over hooks' formula~\cite{georgi_lie_2019}. The factors are numbers assigned to each box of the Young diagram $f(x,y)$ according to the following rule: put $d$ 
in the upper left hand corner of the tableau, and fill the rest of the tableau by adding 1 for each move to the right, and subtracting 1 for each move down. The dimension of $V_{\lambda}$ is given by,
\begin{align}
    d_{\lambda} = \text{dim} V_{\lambda} \equiv \frac{\prod_{(x,y)\in \lambda}f_{x,y}}{\prod_{(x,y)\in \lambda}h_{x,y}}.
\end{align}

For a Young diagram specified by $(p_{1},p_{2},\dots,p_{d-1})$, we have the following formula for $d_{\lambda}$,
\begin{align}\label{eq:dim_irrep}
    d_{\lambda} &=
    \left(p_{1}+1\right) \left(\frac{(p_{1}+p_2+2)(p_2+1)}{2!}\right)\times\nonumber \\
    &\left(\frac{(p_{1}+p_2+p_3+3)(p_2+p_3+2)(p_3+1)}{3!}\right) \dots \nonumber \\
    &= \prod_{s = 1}^{d-1}\frac{\prod_{x = 1}^{s}\left(\sum_{y = 1}^{s-x+1}p_{x+y-1}+s-x+1\right)}{s!}.
\end{align}
The formula for the multiplicity $\mathfrak{C}^{N}_{\lambda}$ is given in terms of a multinomial coefficient of the form ${N \choose {r_1~~r_2~~\dots~~r_d}}$ and $a$ (which is defined as $a = N-\sum_{i = 1}^{d-1}(d-i)p_{i}$),
\begin{align}\label{eq:multiplicity_su_d}
    &\mathfrak{C}^{N}_{\lambda} = \nonumber \\& {N \choose {a~~~(a+p_1)~~~(a+p_1+p_2)~~~\dots~~~ (a+\sum_{i = 1}^{d-1}p_i)}}\times \nonumber \\
    &\biggl[\left(\frac{p_1+1}{a+p_1+1}\right)\left(\frac{(p_1+p_2+2)(p_2+1)}{(a+p_1+p_2+2)(a+p_1+p_2+1)}\right)\dots \biggr]\nonumber\\
    &= \frac{N!}{a!\prod_{s = 1}^{d-1}(a+\sum_{x = 1}^{s}p_x)!}\times \nonumber \\
    &\prod_{t = 1}^{d-1}\frac{\prod_{x = 1}^{t}\left(\sum_{y = 1}^{t-x+1}p_{x+y-1}+t-x+1\right)}{\prod_{x = 1}^{t}\left(a+\sum_{y = 1}^{t}p_{t}+t-x+1\right)}
\end{align}

\subsection{Half-chain entanglement of $SU(d)$ maximally mixed invariant state}
For concreteness, we will consider the maximally mixed $SU(d)$ invariant state on $2N$ qudits, and we will characterise the bipartite entanglement for an equal bipartition, $|A| = |B| = N$. The Schur-Weyl duality allows us to use Eq.~\eqref{eq:mmis_entanglement} to compute the entanglement properties of $SU(d)$ invariant maximally mixed state, as the multiplicities can be explicitly computed using combinatorial techniques, as discussed in the preceding section.

The probability measure $p_{J} = \mathfrak{C}_{[J]}^{N} \mathfrak{C}_{[J^{*}]}^{N} / \mathfrak{C}_{\triv}^{2N}$ is dominated by its maxima in the large $N$ limit, at the saddle point $J_{\text{max}}$ that maximizes the multiplicity $\mathfrak{C}_{[J]}^{N}\mathfrak{C}_{[J^{*}]}^{N}$. Let the Young diagram corresponding to $J$ be $\lambda = (p_{1},p_2, \dots p_{d-1})$. The Young diagram corresponding to the conjugate representation, $J^{*}$ is given by `flipping' $\lambda$, as $\lambda^{*} = (p_{d-1},p_{d-2}, \dots p_{1})$. The dimensions of the $SU(d)$ irreps corresponding to $\lambda$ and $\lambda^{*}$ are equal, as can be confirmed explicitly using Eq.~\eqref{eq:dim_irrep}.

The multiplicities $\mathfrak{C}^{N}_{\lambda}$ and $\mathfrak{C}^{N}_{\lambda^*}$ are given by the expressions in Eq.~\eqref{eq:multiplicity_su_d}. To find the saddle point corresponding to the maximal multiplicity factor, we consider $p_i$ such that $1 \ll p_{i} \ll N$. In this limit, we can rewrite Eq.~\eqref{eq:multiplicity_su_d} as,
\begin{align}\label{eq:multiplicity_asymptotics}
    \mathfrak{C}^{N}_{\lambda = (p_{1}, \dots, p_{d-1})} &\approx A \exp\biggl[\frac{N}{d}\biggl(G(p_{1},\dots,p_{d-1})+\nonumber \\
    &F(p_{1},\dots,p_{d-1})\biggr)\biggr]\left(1+O\left(p^2/N^2\right)\right),
\end{align}
where, $G(p_{1},\dots,p_{d-1})$ is obtained from the mutinomial coefficient factor in Eq.~\eqref{eq:multiplicity_su_d} after using Stirling's approximation,

\begin{align}
    G(p_{1},\dots,p_{d-1}) &= -\sum_{s = 1}^{d-1}H_{s}(p_{1}, \dots, p_{d-1})\times \nonumber \\
    &\log N H_{s}(p_{1}, \dots, p_{d-1}) \text{ where,}\nonumber \\
    H_{s}(p_1, \dots, p_{d-1}) &= 1 - \frac{1}{N}\sum_{t = 1}^{d-1}(d-t)p_{t} + \frac{d}{N}\sum_{t = 1}^{s} p_{t},
\end{align}
and $A$ is a normalization factor.

The rest of the factors in Eq.~\eqref{eq:multiplicity_su_d} give rise to the $F(p_{1},\dots,p_{d-1})$ term after polynomial expansions in the limit $1 \ll p_{i} \ll N$,
\begin{align}
    F(p_1,\dots,p_{d-1}) &= \frac{d}{N}\sum_{s = 1}^{d-1}\sum_{t = 1}^{s}\log \left(\sum_{x = 1}^{s-t+1}p_{s-x+1} \right)
\end{align}

The saddle point equations are $\partial_{i}(G + F) = 0$, where $\partial_{i}G = \frac{\partial G}{\partial p_{i}}$. By an explicit computation, we find,
\begin{align}\label{eq:saddle_eom}
    \partial_{i}(G + F) = &\sum_{s = 1}^{d-1}\sum_{t = 1}^{s}\biggl[\frac{p_{t}}{N^2}\left(d-i-d\Theta_{s,i}\right)+\nonumber \\
    &\frac{1}{N}\frac{\Theta_{s,i}\Theta_{i,x}}{\sum_{x = 1}^{s-t+1}p_{x}} + O\left(p^2/N^2\right)\biggr],\nonumber \\
    \text{where, } \Theta_{i,j} &=\begin{cases}
        1 \text{ if } i\geq j \\
        0 \text{ otherwise}
    \end{cases}
\end{align}
We now use a scaling argument to estimate the saddle point. Suppose we scale $p_{i} \sim \sqrt{N}$; then the saddle point equation of motion Eq.~\eqref{eq:saddle_eom} can be seen to be asymptotically satisfed to $O(1/N)$. This scaling argument also applies to the saddle point of the product $\mathfrak{C}^{N}_{\lambda}\mathfrak{C}^{N}_{\lambda^*}$, which is also maximized for $p_{i}\sim \sqrt{N}$. This allows us to estimate the bipartite entanglement entropy for the maximally mixed $SU(d)$ invariant state to be (using Eq.~\eqref{eq:dim_irrep}), 
\begin{align} \label{eq:ent_sud}
    E_{N:N}\left(\rho_{\triv}^{2N}\right) &\approx \log d_{\lambda_{\text{max}}}|_{p_{i}\sim \sqrt{N}} \approx \log N^{\frac{d(d-1)}{4}} \nonumber \\
    &\approx \frac{d(d-1)}{4}\log N.
\end{align}

\section{MPDO representation of maximally mixed states in fixed charge sectors}
\label{appsec:MPO}

\begin{figure}[t]
    \centering
    \includegraphics[width=0.6\linewidth]{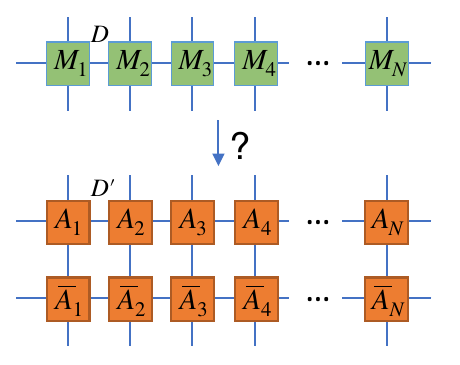}
    \caption{(Top) The MPDO form of a density matrix where $M$ tensors are not necessarily positive semi-definite. 
    (Bottom) The local purification of the state with an MPS bond dimension (purification rank) $D'$. The vertical contracted lines correspond to traced-out ancilla states.}  \label{fig:MPDO} 
\end{figure}

By definition, an MPDO is a matrix product operator representing a density matrix. Any MPDO $\rho$ can be generally viewed as an MPS by reshaping it to the vector $|\rho)$ by the transformation $\rho = \sum \rho_{ij} \ket{i} \bra{j} \to |\rho) = \sum \rho_{ij} \ket{i} \ket{j}$. Hence, the bond dimension of an MPDO, $D$, is determined by the operator entanglement via $\log (D) = O_{A:A^c}(\rho)$.

The operator entanglement of maximally mixed invariant states is given by Eq.~\eqref{eq:op_ent}. For the equal bi-partition of the system into $N/2:N/2$ qudits, the operator entanglement $O_{\frac{N}{2}:\frac{N}{2}}(\rho_{\triv}^N)$ remains constant for finite non-Abelian groups and grows polynomially with the system size for non-Abelian Lie groups in the $N \to \infty$ limit. To show why this is the case, we need to find how the quantum and classical contributions in Eq.~\eqref{eq:op_ent} scale with the system size. We have already discussed the scaling of the quantum entanglement for different groups. Here we explain why the classical contribution behaves similarly.

The classical term in Eq.~\eqref{eq:op_ent} is the Shannon entropy $H[p_J]$ of the probability distribution $p_J = \frac{\mathfrak{C}^{N_A}_{J}\mathfrak{C}^{N_{A^c}}_{J^*}}{\mathfrak{C}_{\triv}^{N}}$. For finite groups the set $\{p_J\}$ is defined on a finite domain of irreps $J$ not scaling with the system size, leading to a constant Shannon entropy in the thermodynamic limit. On the other hand, for semi-simple Lie groups, the probability $p_J$ is localized around $j_{max}$, and we expect the magnitude scaling as $p_{j_{max}} \sim 1/poly(N)$. For example, for $SU(d)$, one can use expression given in Eq.~\eqref{eq:multiplicity_su_d} to show $\mathfrak{C}_{\lambda_{max}} \sim N^{\frac{2-d-d^2}{4}}$ and $\mathfrak{C}_{\triv} \sim N^{\frac{1-d^2}{2}}$, which gives $p_{\lambda_{max}} \sim N^{\frac{1-d}{2}}$ resulting in $H[P_\lambda] \approx -\log(p_{\lambda_{max}}) \sim \frac{(d-1)}{2} \log(N)$. Combining this with Eq.~\eqref{eq:ent_sud}, for the $SU(d)$ MMIS, we find 
\begin{eqnarray}
   \log D \approx \frac{d^2-1}{2} \log(N)  
\end{eqnarray}
in the thermodynamic limit. This result implies that MMISs can be represented as MPOs with bond dimensions no larger than polynomials, enabling us to numerically simulate such states for relatively large system sizes.

We finish this Appendix by briefly discussing purification of maximally mixed symmetric MPDOs into matrix product states. Any MPDO $\rho$ can be locally purified to an MPS $\ket{\psi}$ such that $\rho = \Tr_A \ket{\psi} \bra{\psi}$, where the trace is taken over a set of local ancilla qudits as illustrated in Fig.~\ref{fig:MPDO}. The purified representation of a density matrix has some advantages. For example, the positivity of the density matrix is satisfied in a local way \cite{cirac_MPDO_first}. One natural question is to ask how the minimal bond dimension of the purified state, $D'$ (also known as purification rank), can be bounded by the MPDO bond dimension $D$. It is shown in Ref.~\cite{cirac_MPDO_BondDim} by a concrete example that there exist some MDPOs with small $D$ but unboundedly large $D'$. This shows that small MPDO bond dimension does not necessarily imply small $D'$. Here we ask the same question for maximally mixed symmetric states: how is the purification rank of such states related to their MDPO bond dimension discussed earlier? Ref.~\cite{cirac_MPDO_BondDim} also provides the following inequality which is true for any MPDO and its purifications. \begin{eqnarray}
    \sqrt{D} \leq D' \leq \frac{D^m - 1}{D - 1},
\end{eqnarray}
where $m$ is the number of different eigenvalues of $\rho$ (including zeros). For a general density matrix, $m$ can be as large as Hilbert space dimension, leading to a loose upper bound for $D'$. However, for maximally mixed symmetric states, $m = 2$, since the density matrix is proportional to a projection operator. Therefore,
\begin{eqnarray}
    \sqrt{D} \leq D' \leq D+1.
\end{eqnarray}

There is also a more direct way to see why the second inequality above has to be true for states $\rho$ that are equiprobable mixtures of orthogonal pure states, i.e. 
\begin{eqnarray}
    \rho = \frac{1}{N} \sum_i \ket{e_i} \bra{e_i}.
\end{eqnarray}
We note again that the MPDO bond dimension $D$ is the bond dimension of the double state $| \rho ) = \frac{1}{\sqrt{N}} \sum \ket{e_i} \ket{e_i}$. Such $| \rho )$, however, is also a valid purification of the density matrix $\rho$, since $\Tr_A  | \rho ) ( \rho | = \rho$. This immediately implies $D' \leq D$.

Applying the results to the cases we have studied, we find that for the case of continuous non-Abelian groups the purification rank scales no lower than polynomial and for discrete non-Abelian groups, there always exists a purification of finite rank not scaling with the system size.

\bibliography{resubmit.bib}

\end{document}

%% file: firstfig.pdf_tex
\begingroup%
  \makeatletter%
  \providecommand\color[2][]{%
    \errmessage{(Inkscape) Color is used for the text in Inkscape, but the package 'color.sty' is not loaded}%
    \renewcommand\color[2][]{}%
  }%
  \providecommand\transparent[1]{%
    \errmessage{(Inkscape) Transparency is used (non-zero) for the text in Inkscape, but the package 'transparent.sty' is not loaded}%
    \renewcommand\transparent[1]{}%
  }%
  \providecommand\rotatebox[2]{#2}%
  \newcommand*\fsize{\dimexpr\f@size pt\relax}%
  \newcommand*\lineheight[1]{\fontsize{\fsize}{#1\fsize}\selectfont}%
  \ifx\svgwidth\undefined%
    \setlength{\unitlength}{422.37750244bp}%
    \ifx\svgscale\undefined%
      \relax%
    \else%
      \setlength{\unitlength}{\unitlength * \real{\svgscale}}%
    \fi%
  \else%
    \setlength{\unitlength}{\svgwidth}%
  \fi%
  \global\let\svgwidth\undefined%
  \global\let\svgscale\undefined%
  \makeatother%
  \begin{picture}(1,0.25582715)%
    \lineheight{1}%
    \setlength\tabcolsep{0pt}%
    \put(0,0){\includegraphics[width=\unitlength,page=1]{firstfig.pdf}}%
    \put(0.33595503,0.11468676){\color[rgb]{0,0,0}\makebox(0,0)[t]{\lineheight{0.80000001}\smash{\begin{tabular}[t]{c}\Large $\rho_\triv$\end{tabular}}}}%
    \put(0,0){\includegraphics[width=\unitlength,page=2]{firstfig.pdf}}%
    \put(0.01553113,0.11240209){\color[rgb]{0.12156863,0.02745098,1}\makebox(0,0)[t]{\lineheight{0.80000001}\smash{\begin{tabular}[t]{c}$N$\end{tabular}}}}%
    \put(0,0){\includegraphics[width=\unitlength,page=3]{firstfig.pdf}}%
    \put(0.17814273,0.11468676){\color[rgb]{0,0,0}\makebox(0,0)[t]{\lineheight{0.80000001}\smash{\begin{tabular}[t]{c}$\cdots$\end{tabular}}}}%
    \put(0,0){\includegraphics[width=\unitlength,page=4]{firstfig.pdf}}%
    \put(0.28343292,0.23944683){\color[rgb]{1,0.02745098,0.02745098}\makebox(0,0)[t]{\lineheight{0.80000001}\smash{\begin{tabular}[t]{c}$T$\end{tabular}}}}%
    \put(0,0){\includegraphics[width=\unitlength,page=5]{firstfig.pdf}}%
    \put(0.48145161,0.23144004){\color[rgb]{0,0,0}\makebox(0,0)[t]{\lineheight{0.80000001}\smash{\begin{tabular}[t]{c}\textbf{Result \ref{res:SS}}\end{tabular}}}}%
    \put(0,0){\includegraphics[width=\unitlength,page=6]{firstfig.pdf}}%
    \put(0.47950442,0.20871308){\color[rgb]{0,0,0}\makebox(0,0)[t]{\lineheight{0.80000001}\smash{\begin{tabular}[t]{c}\footnotesize $\rho_J$ is the unique\\\footnotesize steady state of \\\footnotesize strongly symmetric and\\\footnotesize algebraically complete\\\footnotesize unital channels\end{tabular}}}}%
    \put(0.48036339,0.11183891){\color[rgb]{0,0,0}\makebox(0,0)[t]{\lineheight{0.80000001}\smash{\begin{tabular}[t]{c}\textbf{Result \ref{res:Ent}}\end{tabular}}}}%
    \put(0.7915688,0.23216815){\color[rgb]{0,0,0}\makebox(0,0)[t]{\lineheight{0.80000001}\smash{\begin{tabular}[t]{c}\textbf{Result \ref{res:EBC}}\end{tabular}}}}%
    \put(0,0){\includegraphics[width=\unitlength,page=7]{firstfig.pdf}}%
    \put(0.48024265,0.0875045){\color[rgb]{0,0,0}\makebox(0,0)[t]{\lineheight{0.80000001}\smash{\begin{tabular}[t]{c}$E^D=E=E^F$\end{tabular}}}}%
    \put(0.07239663,0.23944683){\color[rgb]{1,0.02745098,0.02745098}\makebox(0,0)[t]{\lineheight{0.80000001}\smash{\begin{tabular}[t]{c}$t=0$\end{tabular}}}}%
    \put(0,0){\includegraphics[width=\unitlength,page=8]{firstfig.pdf}}%
    \put(0.76002066,0.15442001){\color[rgb]{0.9254902,0.54117647,0}\makebox(0,0)[rt]{\lineheight{0.80000001}\smash{\begin{tabular}[t]{r}\scriptsize 0\end{tabular}}}}%
    \put(0.76019782,0.09129209){\color[rgb]{0.9254902,0.54117647,0}\makebox(0,0)[rt]{\lineheight{0.80000001}\smash{\begin{tabular}[t]{r}\scriptsize $O(1)$\end{tabular}}}}%
    \put(0.76032629,0.04712822){\color[rgb]{0.9254902,0.54117647,0}\makebox(0,0)[rt]{\lineheight{0.80000001}\smash{\begin{tabular}[t]{r}\scriptsize $O(\log N)$\end{tabular}}}}%
    \put(0,0){\includegraphics[width=\unitlength,page=9]{firstfig.pdf}}%
    \put(0.92009511,0.23216815){\color[rgb]{0,0,0}\makebox(0,0)[t]{\lineheight{0.80000001}\smash{\begin{tabular}[t]{c}\textbf{Result \ref{res:time}}\end{tabular}}}}%
    \put(0.92010935,0.1770023){\color[rgb]{1,0.02745098,0.02745098}\makebox(0,0)[t]{\lineheight{0.80000001}\smash{\begin{tabular}[t]{c}\small $T = \Omega(1)$\end{tabular}}}}%
    \put(0.92010935,0.10076178){\color[rgb]{1,0.02745098,0.02745098}\makebox(0,0)[t]{\lineheight{0.80000001}\smash{\begin{tabular}[t]{c}\small $T = \Omega(1)$\end{tabular}}}}%
    \put(0.92010935,0.03833435){\color[rgb]{1,0.02745098,0.02745098}\makebox(0,0)[t]{\lineheight{0.80000001}\smash{\begin{tabular}[t]{c}\small $T = \Omega(N)$\end{tabular}}}}%
    \put(0.92010935,0.02193176){\color[rgb]{1,0.02745098,0.02745098}\makebox(0,0)[t]{\lineheight{0.80000001}\smash{\begin{tabular}[t]{c}\tiny $T_{\text{adap}} = \Omega(\log N)$\end{tabular}}}}%
    \put(0,0){\includegraphics[width=\unitlength,page=10]{firstfig.pdf}}%
    \put(0.48140898,0.04861339){\color[rgb]{0,0,0}\makebox(0,0)[t]{\lineheight{0.80000001}\smash{\begin{tabular}[t]{c}\textbf{Result \ref{res:exp_values}}\end{tabular}}}}%
    \put(0.48128825,0.0213122){\color[rgb]{0,0,0}\makebox(0,0)[t]{\lineheight{0.80000001}\smash{\begin{tabular}[t]{c}$\Tr_{A^c}[\rho_\triv] \xrightarrow[|A| = O(1)]{N\to \infty} \mathbb{I}_A$\end{tabular}}}}%
    \put(0.76888648,0.20208807){\color[rgb]{0,0,0}\makebox(0,0)[lt]{\lineheight{0.80000001}\smash{\begin{tabular}[t]{l}\tiny Entanglement$^*$ \end{tabular}}}}%
    \put(0.64415325,0.23216815){\color[rgb]{0,0,0}\makebox(0,0)[t]{\lineheight{0.80000001}\smash{\begin{tabular}[t]{c}Symmetry\end{tabular}}}}%
    \put(0.64433084,0.17765322){\color[rgb]{0,0,0}\makebox(0,0)[t]{\lineheight{0.80000001}\smash{\begin{tabular}[t]{c}Abelian\end{tabular}}}}%
    \put(0.64437347,0.1110824){\color[rgb]{0,0,0}\makebox(0,0)[t]{\lineheight{0.80000001}\smash{\begin{tabular}[t]{c}Non-Abelian\\Discrete\end{tabular}}}}%
    \put(0.64425982,0.03950711){\color[rgb]{0,0,0}\makebox(0,0)[t]{\lineheight{0.80000001}\smash{\begin{tabular}[t]{c}Non-Abelian\\Continuous$^\dagger$\end{tabular}}}}%
    \put(0.8443499,0.1544129){\color[rgb]{0,0,0}\makebox(0,0)[lt]{\lineheight{0.80000001}\smash{\begin{tabular}[t]{l}\tiny $|A|$\end{tabular}}}}%
    \put(0,0){\includegraphics[width=\unitlength,page=11]{firstfig.pdf}}%
    \put(0.82790783,0.14439225){\color[rgb]{0,0,0}\makebox(0,0)[t]{\lineheight{0.80000001}\smash{\begin{tabular}[t]{c}\tiny $N$\end{tabular}}}}%
    \put(0.02257903,0.23025918){\makebox(0,0)[t]{\lineheight{1.25}\smash{\begin{tabular}[t]{c}(a)\end{tabular}}}}%
    \put(0.36282901,0.23025918){\makebox(0,0)[t]{\lineheight{1.25}\smash{\begin{tabular}[t]{c}(b)\end{tabular}}}}%
    \put(0,0){\includegraphics[width=\unitlength,page=12]{firstfig.pdf}}%
    \put(0.09186109,0.00459869){\makebox(0,0)[lt]{\lineheight{1.25}\smash{\begin{tabular}[t]{l}\small Symmetric gates\end{tabular}}}}%
    \put(0,0){\includegraphics[width=\unitlength,page=13]{firstfig.pdf}}%
    \put(0.3349724,0.09292807){\makebox(0,0)[t]{\lineheight{1.25}\smash{\begin{tabular}[t]{c}\small (MMIS)\end{tabular}}}}%
  \end{picture}%
\endgroup%